\documentclass{article}
\usepackage{graphicx} 

\usepackage{todonotes}
\usepackage{amsthm}
\newtheorem{theorem}{Theorem}
\newtheorem{definition}[theorem]{Definition}

\usepackage{geometry}
\usepackage[utf8]{inputenc}
\DeclareUnicodeCharacter{2212}{-}
\usepackage{amsmath}
\usepackage{microtype}
\usepackage{siunitx}
\usepackage{authblk}
\usepackage{lscape}
\usepackage{graphicx}
\usepackage{subcaption}
\usepackage[font={footnotesize},labelfont=bf]{caption}
\usepackage[numbers,sort&compress]{natbib}
\usepackage{xr}
\usepackage{booktabs}
\usepackage{tabularx}
\usepackage{blindtext}
\usepackage{booktabs,multirow,array}
\usepackage{tikz}
\usepackage{verbatim}
\usepackage{chemfig}
\usepackage{gensymb}
\usepackage{textcomp}
\usepackage{setspace}
\usepackage{multicol}
\usepackage{filecontents,pgfplots,pgfplotstable}
\usepackage{afterpage}
\usepackage{xtab,afterpage}
\usepackage{amssymb}
\usepackage{rotating}
\usepackage{scrextend}
\usepackage{tcolorbox}
\usepackage{multirow}
\usepackage{dirtree}
\usepackage{calc}
\usepackage{algorithm}
\usepackage[noend]{algpseudocode}
\usepackage{enumitem}
\setlist[itemize]{leftmargin=*}
\usepackage{microtype}
\usepackage{float}
\usepackage{forest}
\usepackage{footmisc}
\usepackage[hidelinks]{hyperref}
\usepackage{cleveref}
\pgfplotsset{compat=1.18}

\algnewcommand\algorithmicinput{\textbf{Preprocessing:}}
\algnewcommand\Preprocessing{\item[\algorithmicinput]}
\algnewcommand\algorithmicequire{\textbf{Note:}}
\algnewcommand\Note{\item[\algorithmicequire]}

\algrenewcommand\algorithmicrequire{\textbf{Input:}}

\usepackage{arxiv}
\begin{document}
\title{Prediction of microstructural representativity from a single image}

\author[1]{{Amir Dahari}
 \ }

\author[1, 2]{{Ronan Docherty}
 \ }

\author[1]{{Steve Kench}
 \ }

\author[1]{Samuel J. Cooper \ \ }

\affil[1]{{\textit{\footnotesize Dyson School of Design Engineering, Imperial College London, London SW7 2DB}}}
\affil[2]{{\textit{\footnotesize Department of Materials, Imperial College London, London SW7 2AZ}}}

\maketitle




\begin{abstract}
\begin{center}
\begin{minipage}{0.85\textwidth}
{\small In this study, we present a method for predicting the representativity of the phase fraction observed in a single image (2D or 3D) of a material. Traditional approaches often require large datasets and extensive statistical analysis to estimate the Integral Range, a key factor in determining the variance of microstructural properties. Our method leverages the Two-Point Correlation function to directly estimate the variance from a single image, thereby enabling phase fraction prediction with associated confidence levels. We validate our approach using open-source datasets, demonstrating its efficacy across diverse microstructures. This technique significantly reduces the data requirements for representativity analysis, providing a practical tool for material scientists and engineers working with limited microstructural data. To make the method easily accessible, we have created a web-application, \url{www.imagerep.io}, for quick, simple and informative use of the method.}
\end{minipage}
\end{center}
\end{abstract}
\vspace{.2cm}
\newpage

\tableofcontents

\newpage

\section{Introduction}

Microscopy is an important aspect of materials characterisation, allowing the statistical analysis and modelling of material microstructure. Through techniques like scanning electron microscopy (SEM) and X-ray computed tomography (XCT), researchers can obtain detailed images of the microstructure, which form the basis for understanding the material’s properties and behaviour. This approach is critical to the design of a wide range of technologies, including efficient battery electrodes, corrosion resistant alloys, and bio-compatible bone implants. However, a fundamental challenge in using microscopy for material characterization is assessing the representativity of the data extracted from a single image.

In many materials even slight variations in phase fraction can have a significant impact on the material’s overall properties, but representativity analysis of this metric is rare in the literature. This is perhaps because traditional methods of assessing representativity rely on large datasets, potentially requiring multiple image acquisitions, which are often resource-intensive. Furthermore, the computational expense of modelling techniques scales with the number of pixels/voxels; thus, it may sometimes be undesirable to discover that your technique requires an order of magnitude more data to give reasonable confidence in the results. To this end, we often choose to simply bury our heads in the sand, and report results with no thought given to sample representativity.

In this paper, we propose a simple model, ImageRep, for estimating confidence in phase fraction from a single micrograph or microstructure (2D or 3D). While most of our model is analytical, we use the MicroLib dataset \cite{kench2022microlib} to support one specific, data-driven component: estimating uncertainty in the variance prediction of the phase fraction. Each Microlib entry includes a trained generative adversarial network that enables synthesis of microstructure samples of any size. We leverage this capability to construct a model that directly relates the two-point correlation function (TPC) of an image to the confidence level in representativity of the phase fraction observed in that image. 

Crucially, the fast approximation using the TPC can be calculated on a single micrograph or microstructure, giving users easy access to an estimate of phase fraction confidence. This also allows prediction of the image size required to obtain a user specified uncertainty in phase fraction, which can direct precise data collection campaigns after a small initial data collection study. 

The connection between phase fraction variation and the TPC function was originally proposed by Matheron in \cite{metheron1971theory} and later presented and discussed in \cite{kanit2003determination, matheron2012estimating, ohser2009spectral, chiu2013sec64, jeulin2021sec355}. Across this literature, the equation connecting the TPC function and the phase fraction variation (formulated with different notations but fundamentally the same) is presented as a theoretical result, as it requires the true, yet unknown, bulk material phase fraction as an input. Consequently, it has not been used to directly predict phase fraction variation from limited data. In this paper, we show how this equation can be modified for practical use, without requiring knowledge of the true material phase fraction. Importantly, we prove that, in expectation, our model's prediction exactly matches the phase fraction variation. 

Thus, our contribution is twofold: first, we improve the theoretical understanding of the connection between the TPC function and phase fraction variation as presented in the existing literature; and second, we leverage this insight to develop a practical method for predicting phase fraction variation and for better assessing image representativity. 

Although we prove that our variance prediction is correct in expectation, a prediction of phase fraction variation based on a single image, representing only a small sample of the material, is still subject to uncertainty. This is because the TPC function itself depends on the representativity of the image. Specifically, there is variability in the predicted variation itself, which can be understood as the variance of the phase fraction variance. We approximate this discrepancy using a data-driven approach based on the diverse set of materials available in MicroLib \cite{kench2022microlib}, and we incorporate it into our final representativity prediction, as described in Section \ref{sec:predicting the confidence level}. 

We then validate our model using three open-source battery materials datasets \cite{hsu2018mesoscale, lagadec2016communication} as well as a rich library of synthetic microstructures generated with the PoreSpy microstructure generators \cite{gostick2019porespy}. We show that this last data-driven correction has a small but positive effect. Finally, to make our ImageRep model easily accessible to microscopists and materials scientists, we have created a web application at \url{www.imagerep.io} for quick, simple and informative use of the model. 

\begin{figure}[!htb]
    
    \includegraphics[width=1\textwidth]{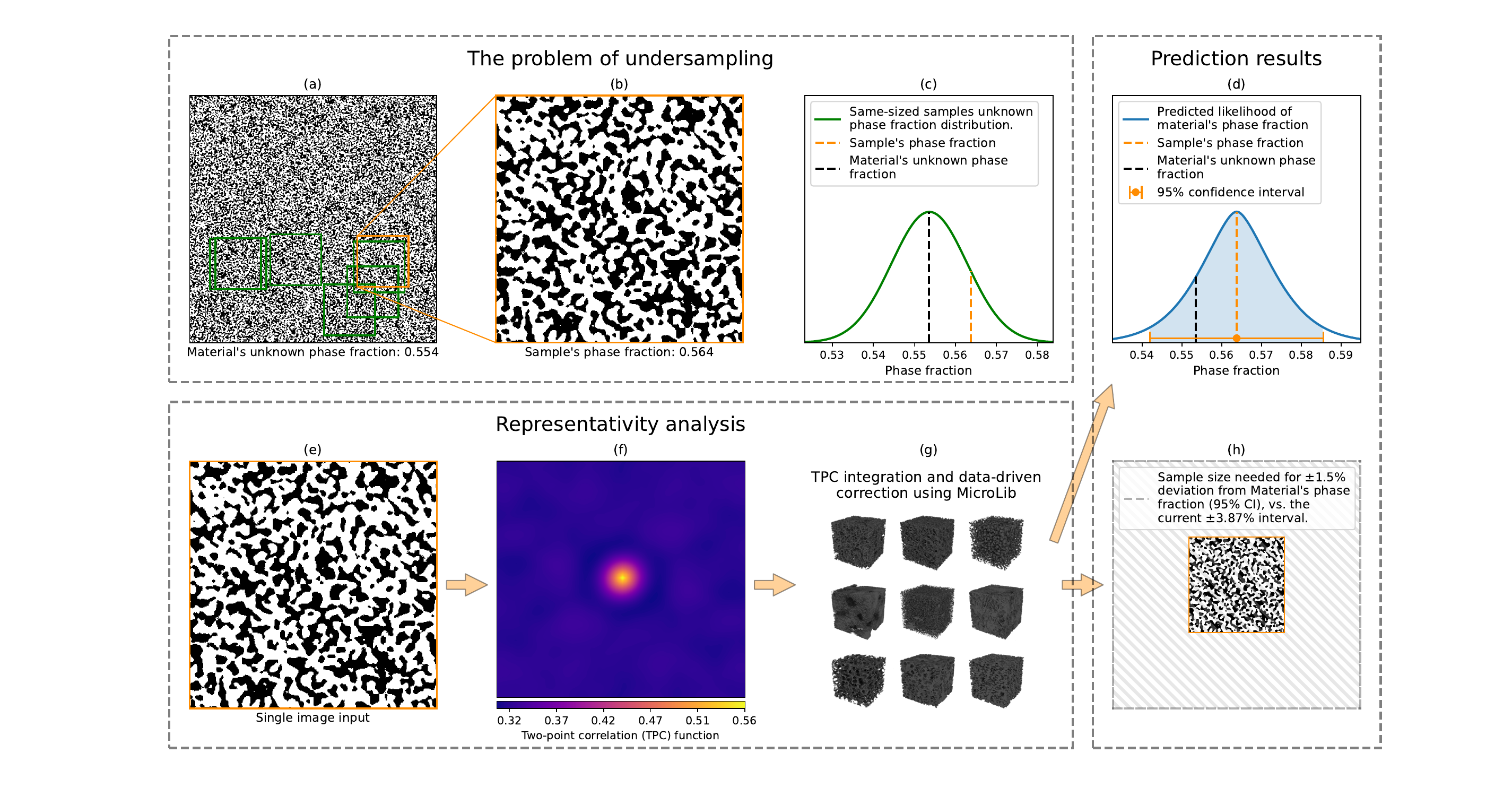}
    \centering
    \caption{An outline of the undersampling problem and the solution presented in this paper. Starting from (a) a large ($1697\times1697$ pixels) image of a SOFC anode \cite{hsu2018mesoscale}, we assume that a (b) small sample ($350\times350$ pixels) is imaged. The phase fraction of the large image is simplified to that of the bulk material\protect\footnotemark, while only the sample phase fraction is observed. Since the difference between sample and bulk phase fractions is unknown, confidence bounds are needed to assess how reliably the sample represents the material. (c) shows the fitted distribution for the histogram of phase fractions of all $350\times350$ pixels image samples taken from (a) (which is also unknown to the practitioner), together with the observed sample phase fraction highlighted. The representativity analysis input is a single image (e) (the same sample shown in (b)). The TPC function of (e) is calculated and its shortest $40\times40$ distances are presented in (f). By utilizing both the TPC and a data-driven analysis of MicroLib \cite{kench2022microlib} in a statistical framework presented in Section \ref{sec:model theory} (g), both the confidence bounds to the material's phase fraction (d) and the image size needed for a specified deviation from the material's phase fraction (h) are determined. 
    }
    \label{fig:introduction}
\end{figure}

The theoretical framework behind our model is developed in Sections \ref{sec: Section 2: representativity prediction} and \ref{sec:model theory}. These sections are technical in nature, so readers focused on empirical validation and implications may prefer to continue to the Validation or Discussion sections \ref{sec: Validation} and \ref{sec: Discussion}.

\section{Representativity Prediction}\label{sec: Section 2: representativity prediction}

\footnotetext{For simplicity, the phase fraction of Figure \ref{fig:introduction}(a) is regarded as the material's phase fraction, although the image has its own uncertainty, since it is still a small sample from the material. According to our method, the image phase fraction deviates by at most 0.81\%  from the true material's phase fraction, with 95\% confidence.\label{footnote:uncertainty in sofc large image}}
\subsection{Representativity}

For a truly heterogeneous material (i.e. non-uniform and non-periodic) a finite sample can never be considered to be \textit{fully} representative, but simply within an acceptable tolerance of likelihood of representing the material \cite{drugan1996micromechanics}. To put it another way, the true value of any property of a material cannot be directly measured, but samples of increasing size are more likely to be closer to this value. This concept applies equally to data of arbitrary dimension, but is most commonly applied to volumes (3D data) and images (2D data) in materials science.  

Representativity is also property dependant, so a sample could be more representative in one metric (e.g. phase fraction) than another (e.g. interfacial surface area). Furthermore, representativity can be process dependent, in the sense that emergent structures of a process, for example turbulent flow in porous media, might require even larger volumes to be analysed, but this aspect is beyond the scope of the present investigation.

The following definition takes into account both the stochasticity and property dependency necessary to understand the representativity of a sample.  

\begin{definition}\label{def: abstract definition}
    A sample is $(c, d)$-m representative if the measured metric, m, in the sample deviates by no more than $d\%$ from the bulk material property, with at least $c\%$ confidence.
\end{definition}

In this study we will start by analysing perhaps the simplest example of representativity, which is the phase fraction of a two-phase material in a 2D image. As we will show, for a two-phase material, the absolute error in the phase fraction is the same for both phases, so the result does not depend on which phase is selected for analysis. Furthermore, following the derivation, we will also show that this approach can be trivially generalised to 3D volumes and $n$-phase materials. 

\begin{definition}\label{def: concrete definition}
    An image of a two-phase material is $(c, d)$-phase fraction representative if the measured phase fraction in the image deviates by no more than $d\%$ from the bulk material property, with at least $c\%$ confidence.
    \end{definition}

A practical example from materials science is an SEM image of a dewetting metal thin film on the surface of a ceramic substrate, where the authors wished to measure the area coverage fraction (i.e. phase fraction) of the metal, as described in \cite{song2019unveiling}. Phase fraction representativity is not discussed in the paper, but they may have wanted to know with what confidence they could expect to be within $1\%$ of the true value of the bulk material. 

\subsection{Model derivation}\label{sec: model derivation}

To formalize the setting of Definition \ref{def: concrete definition}, we describe the image domain size 
$X$ as a discrete rectangular coordinate system of size $l_1\times l_2$, where $l_1$ and $l_2$ are the numbers of pixels along each axis. Each pixel $x\in X$ is identified by its coordinate $x=(x_1,x_2)$. Let $\Omega=\Omega_X$ denote the space of all possible images $\omega\in\Omega$ of the material of size $X$.
     
Let $\Phi=\Phi_X:\Omega\rightarrow[0,1]$ be a random variable that calculates the phase fraction of an image. $\Phi(\omega)$ is simply the phase fraction of phase 1 in the image $\omega$ (N.B. if the image was represented with 1s in the pixels containing phase 1 and 0s in the pixels of phase 2, then $\Phi(\omega)$ would simply be the mean of $\omega$'s pixel values). Let the unknown true phase fraction of this material be $\text{E}[\Phi]=\phi$. 

The variation in the phase fraction of a microstructure is strongly related to the size and variability of its features. It turns out that, for materials that have features of finite sizes, as the size of the image from the material grows large, the variance can be explained by one parameter coined the ``integral range", $A_n$ \cite{metheron1971theory, dirrenberger2014towards} ($n=2,3$ for the dimension of the image):

\begin{equation}\label{eq:variation as a finction of cls}
    \text{Var}[\Phi]=\frac{A_n}{|X|}\cdot\phi(1-\phi)
\end{equation}

when $|X|\gg A_n$. This equation highlights that for phase fraction variation prediction, the image can be thought of as $\frac{|X|}{A_n}$ independent random Bernoulli variables with probability $\phi$ to be $1$ \cite{kanit2003determination}. To compare between image sizes and dimensions, this finding has motivated us to work with the length scale of $A_n$ instead, and define the Characteristic Length Scale (CLS), $a_n$, to be

\begin{equation}\label{eq:Var as a function of CLS}
    a_n=A_n^{1/n},\ \  \text{Var}[\Phi]=\frac{a_n^n}{|X|}\cdot\phi(1-\phi)
\end{equation}

The CLS can be viewed as a measure of the spatial features in a material, capturing the typical distance over which correlations between different regions of the microstructure decay. Unfortunately, even in the simple case of random circular particles on a plane, the CLS is neither the radius of the particles nor the inter-particle distance, but a combined quantification of the relationship between all features.

Until now, other than rule-of-thumb estimations \cite{awarke2011quantifying, awarke20123d}, in all papers that quantify representativity, $a_n$ is approximated in the manner visualised in the middle column of Figure \ref{fig:representativity introduction} \cite{kanit2003determination, cadiou2019numerical, el2021numerical}. First, images of increasing size from the material are randomly selected - in our experiment we randomly selected $4000$ samples per image size. Second, the variance or standard deviation of the phase fraction is calculated for each size of the randomly selected samples. Lastly, after obtaining different standard deviation data points for increasing sample sizes, the best integral range is found that fits Equation (\ref{eq:Var as a function of CLS}) \cite{kanit2006apparent}. 

Since $\Phi$ behaves as a mean of independent random Bernoulli variables for large enough $|X|$, it follows from the central limit theorem that we can assume $$\Phi\sim \mathcal{N}(\phi,\sigma^2_X)$$ as shown in \cite{quintanilla1997local}. The normality of $\Phi$ can be used for a confidence interval, to indicate that for a random image $\omega\in\Omega$, with $c\%$ confidence, 

\begin{equation*}
    \Phi(\omega)\in [\phi-\sigma_X\cdot z_{c\%}, \phi+\sigma_X\cdot z_{c\%}]
\end{equation*}

for the appropriate $z$-value at $c\%$ (for example $z_{95\%}\approx 1.96)$. Conversely, since the mean value, $\phi$, is unknown, and only $\Phi(\omega)$ and $\sigma_X$ are known, then with $95\%$ confidence,

\begin{equation}\label{eq:inherent error}
    \phi\in [\Phi(\omega)-\sigma_X\cdot z_{c\%}, \Phi(\omega)+\sigma_X\cdot z_{c\%}]
\end{equation}

For this evaluation of the phase fraction $\phi$, the deviation of $\Phi(\omega)$ from $\phi$ is at most $\sigma_X\cdot z_{c\%}$ with $c\%$ confidence. Therefore, aligning with Definition \ref{def: concrete definition}, if the standard deviation $\sigma_X$ is known, the representativity of an image $\omega$ of size $X$ would be known as well, since the image would be $(c, \frac{\sigma_X\cdot z_{c\%}}{\Phi(\omega)}\cdot100)$-phase fraction representative. Therefore, knowing the representativity of an image is equivalent to knowing the variation of the phase fraction for images of the same size, randomly sampled from the bulk.

\begin{figure}
\centering
\includegraphics[width=0.8\textwidth]{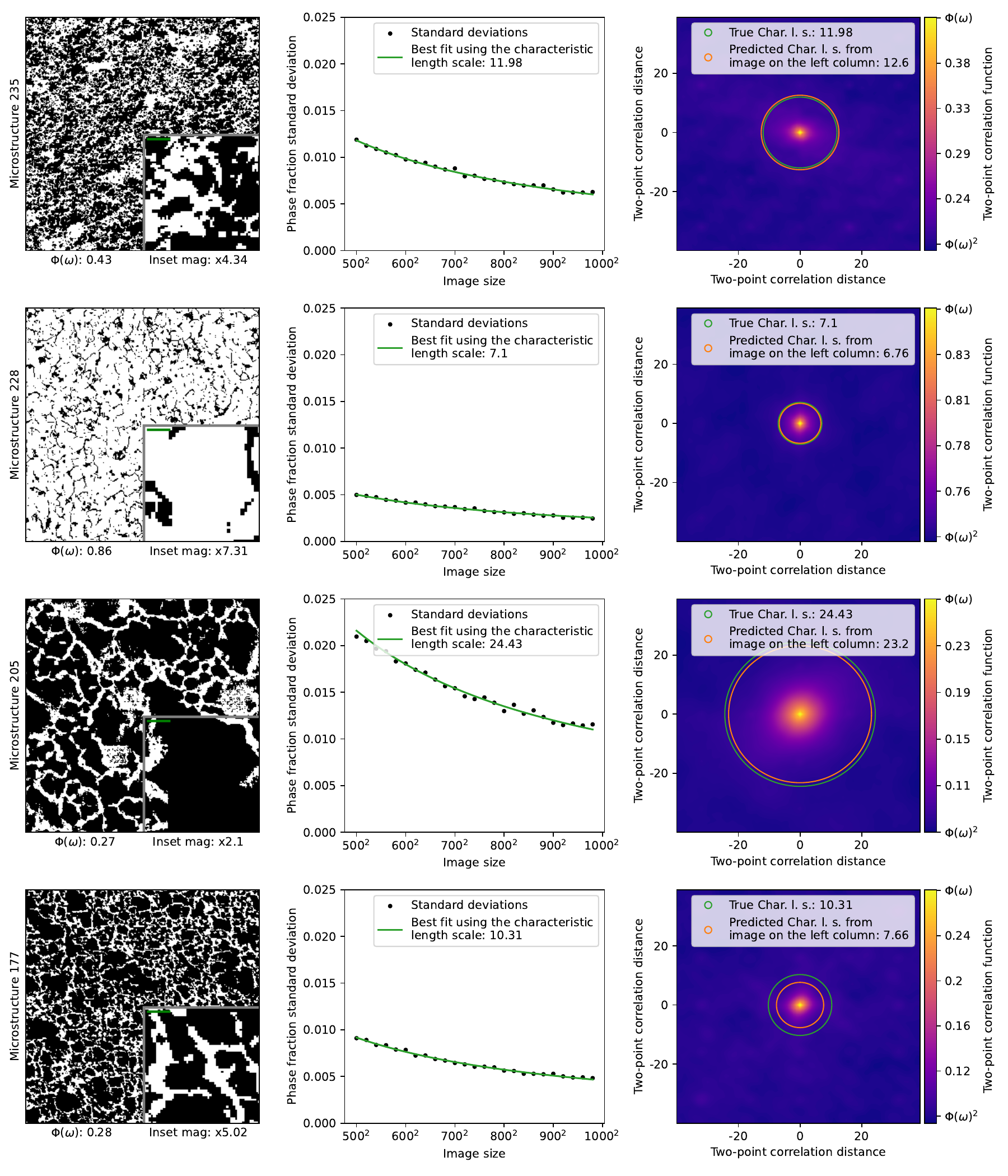}
\caption{Comparison of true and predicted characteristic length scale (CLS), across several MicroLib \cite{kench2022microlib} microstructures. Column 1 shows the microstructure from the MicroLib database, with the size of the characteristic length scale showing as the green line inside the inset. Column 2 shows the trend in standard deviation for increasing sample size of images generated by SliceGAN (4000 images per image size), and a fitted curve calculated by searching for an appropriate CLS using Equation \ref{eq:Var as a function of CLS}. Column 3 shows the two-point correlation (TPC) function of the image, magnified to the $80\times80$ shortest distances, and a comparison of the CLS calculated from the TPC of the image in the left column (presented in this section), versus the aforementioned statistical analysis fitting of thousands of images in the middle column.}
\label{fig:representativity introduction}
\end{figure}

To see the relationship between image size, CLS and the phase fraction standard deviation, see the middle column of Figure \ref{fig:representativity introduction}  and the excellent curve fit of the standard deviations obtained experimentally and from the theoretical variance using the CLS value presented in Equation \ref{eq:Var as a function of CLS}. It is remarkable, after observing similar CLS fits for all of the different microstructures in MicroLib \cite{kench2022microlib}, that one single number $a_n$ embodies the entire feature variability needed to define the microstructure phase fraction variability \cite{matheron2012estimating}. In other words, $a_n$ is the key to defining phase-fraction variance, and is the only variable needed to predict phase fraction representativity. 

As we will show later, applying this approach to a single image by sampling subimages of that image and fitting the CLS, leads to inaccurate predictions. With a single image, a more information-rich method must be used to ensure accurate results.

\subsection{Model overview}

In this subsection, before diving into the theory behind the model, we present a concise overview of the model and its two uses, which will be covered at length in the next section. After obtaining an image of the material, there are two ways for a user to use the model. One way is to predict the image phase fraction deviation from the bulk material phase fraction, with a given confidence level. The other way is to produce a prediction for the image size that would be required in order for the phase fraction to be expected to fall within a particular interval with a specified confidence level. This is the same as stating how much more data needs to be collected to reach a specified degree of representativity. These two uses are presented in the two algorithms below.

\begin{algorithm*}[htb!] 

\begin{algorithmic}[1]
\Require The image $\omega$ (of size $|X|$) and the desired confidence level $c\%$.
\Preprocessing $\sigma_{mod,|X|}\leftarrow$ Obtain the pre-calculated model prediction error for images of size $|X|$, based on the MicroLib dataset, as explained in Figure \ref{fig:model error}.
\Statex
    \State $TPC(\omega)\leftarrow$ Calculate the TPC of the image $\omega\in\Omega_X$. 
    \State $\Tilde{\sigma}\leftarrow$ Estimate the std from the unknown distribution $\Omega_X$ using $TPC(\omega)$, by Equation (\ref{eq:predicting std}). 
    \State Calculate the confidence bounds using $\Tilde{\sigma},c\%,\sigma_{mod,|X|}$  explained in Section \ref{subsec:Total error estimation}.
    \label{State: 3 in algorithm 1}
    
    \Return $d\%$ the deviation of $\Phi(\omega)$ from the bulk material phase fraction, with $c\%$ confidence. 
    \caption{{\footnotesize Prediction of phase fraction representativity of a given image }}
    \label{alg: repres. prediction}
\end{algorithmic}
\end{algorithm*}

\begin{algorithm*}[htb!] 

\begin{algorithmic}[1]
\Require The image $\omega$ (of size $|X|$), desired confidence level $c\%$ and the desired deviation level from the bulk material phase fraction $d\%$.
\Preprocessing $\forall|Y|:\sigma_{mod,|Y|}\leftarrow$ Obtain the pre-calculated model prediction errors fit for different sized images based on the MicroLib dataset, as explained in Figure \ref{fig:model error}.
\Statex
    \State $TPC(\omega)\leftarrow$ Calculate the TPC of the image $\omega\in\Omega_X$. 
    \State $\Tilde{a}_n\leftarrow$ Estimate the CLS  from the unknown distribution $\Omega_X$ using $TPC(\omega)$, by Equation (\ref{eq:model CLS prediction}). 
    \State Using Algorithm 1 (translating $\Tilde{a}_n$ to $\Tilde{\sigma}$ in state \ref{State: 3 in algorithm 1}) , find $|X'|$ for which the returned value is $d\%$. 
    
    \Return $|X'|$ the image size predicted to be $(c,d)$-phase fraction representative, according to Definition \ref{def: concrete definition}.

		\caption{{\footnotesize Prediction of image size needed for a given representativity certainty}}
		\label{alg:image size prediction}
	\end{algorithmic}
\end{algorithm*}

\section{Model Theory}\label{sec:model theory}

\subsection{Variation Prediction}

As discussed in the previous section, knowing the representativity of a random image is equivalent to knowing the sample size-dependent variance of the measured property. This part will be focused on the derivation of the model's prediction of the variance of the phase fraction, and how it has the desired property that in expectation it is precisely the variance of the phase fraction.

Starting off with the same notation as the previous section, if for every location $x\in X$ we set $B_x$ to be a random variable that is 1 if in the image $\omega$ there is phase 1 in position $x$, then the phase fraction of image $\omega$ is simply

\begin{equation*}
    \Phi(\omega)=\frac{1}{|X|}\sum_{x\in X}B_x(\omega)
\end{equation*}

For a given vector $r\in\mathbb{Z}^2$ and image size $X$, the Two-Point Correlation (TPC) function (or Auto-correlation function) $T_r=T_{r,X}:\ \Omega\rightarrow[0,1]$, calculates the mean number of vectors $r$ with both end-points being phase 1 in the image, i.e.

\begin{equation*}
    T_r(\omega)=\frac{1}{|X_r|}\sum_{\substack{{x\in X} \\ {x+r\in X}}} B_x(\omega)\cdot B_{x+r}(\omega)
\end{equation*}
 
with $|X_r|$ defined to be the volume of all points $x\in X$ such that $x+r\in X$ as well. We use the Fast Fourier Transform (FFT) algorithm for fast computation of the TPC as theoretically explained in \cite{adams2012microstructure} and applied in \cite{robertson2022efficient}, reducing the computation time of the TPC of large images from tens of minutes to a few seconds. The connection between the FFT and the TPC is given in the Supplementary Information. 

Before presenting our results, we present the assumption made about the materials examined within the scope of this study. We assume that the materials are homogeneous, meaning their statistical properties, such as the distribution of phases, are the same throughout the material. 

However, homogeneity alone does not imply decay of spatial correlations. Therefore, we introduce a stronger assumption, referred to as the macro-homogeneity assumption, as in prior work \cite{hill1984macroscopic, costanzo2005definitions}. The assumption states that correlations decay with distance, so that regions sufficiently far apart can be treated as independent. I.e. for large vectors $r$, the probability of finding a specific phase in both end-points of the vector converges to the square of the phase fraction of the material, since the endpoints become independent of each other:  
    
\begin{equation*}
    T_r\xrightarrow[|r|\rightarrow \infty]{} \text{E}[\Phi]^2
\end{equation*}

So given a small $\epsilon>0$, there exists a large vector length $r_0$, such that for all $r:\ |r|>r_0$, $|T_r-\text{E}[\Phi]^2|<\epsilon$.

In our model, we assume that for each material, there exists a certain distance above which the values are independent:

\begin{definition}\label{eq:r bigger than r_0}
(Macro-homogeneity assumption) A material is macro-homogeneous if there exists a sufficiently large distance $r_0$ such that $\forall r:\ |r|>r_0:$ $$\normalfont\text{E}[T_r]=\text{E}[\Phi]^2$$ 
\end{definition}

This assumption might simplify the underlying spatial structure of the material, and its implications and limitations are further discussed in the Discussion section.

Using these preliminaries, we can present the result of using the TPC for the variation prediction. Given a single image, there is no way to know the phase fraction of the bulk material, $\text{E}[\Phi]$, but a good estimate of the variance $\text{Var}[\Phi]$ for a specific image size can predict how far the image phase fraction is likely to be from the material phase fraction. Harnessing the TPC in the variance prediction of a single image leads to the following result. We write $\sum_r$ to denote the sum over all vectors $r=x-x'$, where $x,x'\in X$ are coordinates in the image domain.
 
\begin{theorem}\label{theorem: expectation is variance}
    For large enough $r_0$ and appropriate constant $C_{r_0}$, if we set $\Psi:\Omega\rightarrow\mathbb{R}$ to be
    $$\Psi(\omega)=\frac{1}{C_{r_0}}\cdot\sum_{\substack{r \\ |r|\leq r_0}}\frac{|X_r|}{|X|}\left(T_r(\omega)-\Phi^2(\omega)\right)$$ then \normalfont$\text{E}[\Psi]=\text{Var}[\Phi]$.
\end{theorem}

\begin{proof}
    
First, we focus on $\Phi^2(\omega)$:
\begin{equation*}
    \Phi^2(\omega)=\frac{1}{|X|}\sum_{x\in X}B_x(\omega)\cdot\frac{1}{|X|}\sum_{x'\in X}B_{x'}(\omega)=
\end{equation*}
 
\begin{equation*}
    =\frac{1}{|X|^2}\sum_{x\in X}\sum_{x'\in X}B_x(\omega)\cdot B_{x'}(\omega)=
\end{equation*}
 
\begin{equation*}
    =\frac{1}{|X|^2}\sum_r\sum_{\substack{{x\in X} \\ {x+r\in X}}}B_x(\omega)\cdot B_{x+r}(\omega)=
\end{equation*}
Where the last equality is easily seen by double inclusion.
\begin{equation*}
    =\frac{1}{|X|}\sum_r\frac{1}{|X|}\sum_{\substack{{x\in X} \\ {x+r\in X}}}B_x(\omega)\cdot B_{x+r}(\omega)=
\end{equation*}
 
\begin{equation*}
    =\frac{1}{|X|}\sum_r\frac{|X_r|}{|X|}\frac{1}{|X_r|}\sum_{\substack{{x\in X} \\ {x+r\in X}}}B_x(\omega)\cdot B_{x+r}(\omega)=\frac{1}{|X|}\sum_r\frac{|X_r|}{|X|}T_r(\omega)
\end{equation*}

So for an image $\omega$,

\begin{equation*}\label{eq:equality of tpc and observed phase fraction squared}
    \Phi^2(\omega)=\frac{1}{|X|}\sum_r\frac{|X_r|}{|X|}T_r(\omega)
\end{equation*}
 
And from the linearity of the expectation we obtain
\begin{equation}\label{eq: expectation of Phi squared}
    \text{E}[\Phi^2]=\frac{1}{|X|}\sum_r\frac{|X_r|}{|X|}\text{E}[T_r]
\end{equation}

 
Additionally, it's a short exercise to see that 
\begin{equation*}
    \frac{1}{|X|}\sum_r\frac{|X_r|}{|X|}=1
\end{equation*}
 
And so 

\begin{equation}\label{eq: multiplication by 1}
    \text{E}[\Phi^2]=\text{E}[\Phi^2]\cdot\frac{1}{ |X|}\sum_r\frac{|X_r|}{ |X|}=\frac{1}{ |X|}\sum_r\frac{|X_r|}{ |X|}\text{E}[\Phi^2]
\end{equation}

    Combining the macro-homogeneity assumption \ref{eq:r bigger than r_0} and Equations (\ref{eq: expectation of Phi squared}) and (\ref{eq: multiplication by 1}), we obtain
    \begin{equation*}
    \frac{1}{ |X|}\sum_r\frac{|X_r|}{ |X|} \text{E}[\Phi^2]=\text{E}[\Phi^2]=
    \frac{1}{ |X|}\sum_r\frac{|X_r|}{ |X|}\text{E}[T_r]=
\end{equation*}
\begin{equation*}
    =\frac{1}{ |X|}\left(\sum_{\substack{r \\ |r|\leq r_0}}\frac{|X_r|}{ |X|}\text{E}[T_r]+\sum_{\substack{r \\ |r|> r_0}}\frac{|X_r|}{ |X|}\text{E}[T_r]\right)=
\end{equation*}
\begin{equation*}
    =\frac{1}{|X|}\left(\sum_{\substack{r \\|r|\leq r_0}}\frac{|X_r|}{|X|}\text{E}[T_r]+\sum_{\substack{r \\|r|> r_0}}\frac{|X_r|}{|X|}\text{E}[\Phi]^2\right)
\end{equation*}
Multiplying by $|X|$ and subtracting elements from both sides, we reach
\begin{equation}\label{eq:before dividing by normalisation}
    \sum_{\substack{r \\ |r|> r_0}}\frac{|X_r|}{|X|}\left(\text{E}[\Phi^2]-\text{E}[\Phi]^2\right)=\sum_{\substack{r \\ |r|\leq r_0}}\frac{|X_r|}{|X|}\left(\text{E}[T_r]-\text{E}[\Phi^2]\right)
\end{equation}

    If we define the constant $C_{r_0}$ to be
\begin{equation*}
    C_{r_0}=\sum_{\substack{r \\ |r|> r_0}}\frac{|X_r|}{|X|}
\end{equation*}

then since $\text{Var}[\Phi]=\text{E}[\Phi^2]-\text{E}[\Phi]^2$, from Equation (\ref{eq:before dividing by normalisation}) we reach the result
\begin{equation}\label{eq: Variance after taking the expectation}
    \text{Var}[\Phi]=\frac{1}{C_{r_0}}\cdot\sum_{\substack{r \\ |r|\leq r_0}}\frac{|X_r|}{|X|}\left(\text{E}[T_r]-\text{E}[\Phi^2]\right)
\end{equation}

So, if we set the new random variable $\Psi:\Omega\rightarrow\mathbb{R}$ to be

\begin{equation}\label{eq:Psi definition}
    \Psi(\omega)=\frac{1}{C_{r_0}}\cdot\sum_{\substack{r \\ |r|\leq r_0}}\frac{|X_r|}{|X|}\left(T_r(\omega)-\Phi^2(\omega)\right)
\end{equation}

Note that $\Psi$ only uses available attributes of the image. Applying the linearity of the expectation to Equation (\ref{eq:Psi definition}) yields Equation (\ref{eq: Variance after taking the expectation}), and we reach the desired result:

\begin{equation*}
    \text{E}[\Psi]=\text{Var}[\Phi]
\end{equation*}
\end{proof}

Ultimately, we can use $\Psi$ for our model variance prediction. A visual explanation of the proof of Theorem \ref{theorem: expectation is variance} in the simple case of randomly distributed circles on a plane can be seen in Figure \ref{fig:circles explanation}. 

The equations relating the two-point correlation (TPC) function to phase fraction variation, as presented in previous literature \cite{metheron1971theory, kanit2003determination, matheron2012estimating, ohser2009spectral, chiu2013sec64, jeulin2021sec355}, are inherently theoretical due to their reliance on the true but unknown bulk material phase fraction. Among these studies, only \cite{ohser2009spectral} explicitly provides the correct normalization, yet even this formulation depends on the unknown bulk phase fraction. In these equations, replacing the true bulk phase fraction with the image-derived phase fraction would incorrectly yield zero variation, as explicitly shown in equation \ref{eq: expectation of Phi squared}. Consequently, our formulation is the first to enable the prediction of phase fraction variation directly from the TPC function.

The selection of $r_0$ in practice is also presented in the supplementary information. In essence, we measure increasingly large rings from the center of the TPC function and select $r_0$ at the point where the TPC function plateaus and stabilizes. This plateau reflects the transition from correlated to uncorrelated spatial behavior, offering a practical and interpretable criterion. While this approach has proven effective in our analyses, we acknowledge that the choice of $r_0$ remains an important parameter and may benefit from further refinement in future work, which is further explored in the discussion section.

To support the $r_0$ selection and enable an efficient computation via the FFT algorithm \cite{adams2012microstructure}, we use periodic TPC calculation instead of non-periodic TPC.  These adjustments, along with minor modifications to Equation (\ref{eq:Psi definition}) that improve prediction stability, are described in detail in the Supplementary Information, including a discussion of how they preserve the result of Theorem \ref{theorem: expectation is variance}. While non-periodic TPC can be computed by padding the image with zeros and appropriately normalizing, this approach is both slower and more memory-intensive. In contrast, the periodic TPC introduces uncorrelated TPC information, especially at the boundaries, which has a stabilizing effect on the TPC curve. This added stability facilitates the identification of the plateau region and thus supports a more robust selection of $r_0$.

We have presented the case of 2D for simplicity, but all of the calculations also hold for 3D. The only requirement for the extension to 3D is to change the coordinate system $X$ at the beginning of Section \ref{sec: model derivation} to $X = X_{l_1}\times X_{l_2}\times X_{l_3}$. In addition, the only modification required to analyse a material with more than two phases is to treat each phase separately, where pixels of the phase in question are set to one and all other pixels are set to zero.

\begin{figure}
    \centering
    \includegraphics[width=\textwidth]{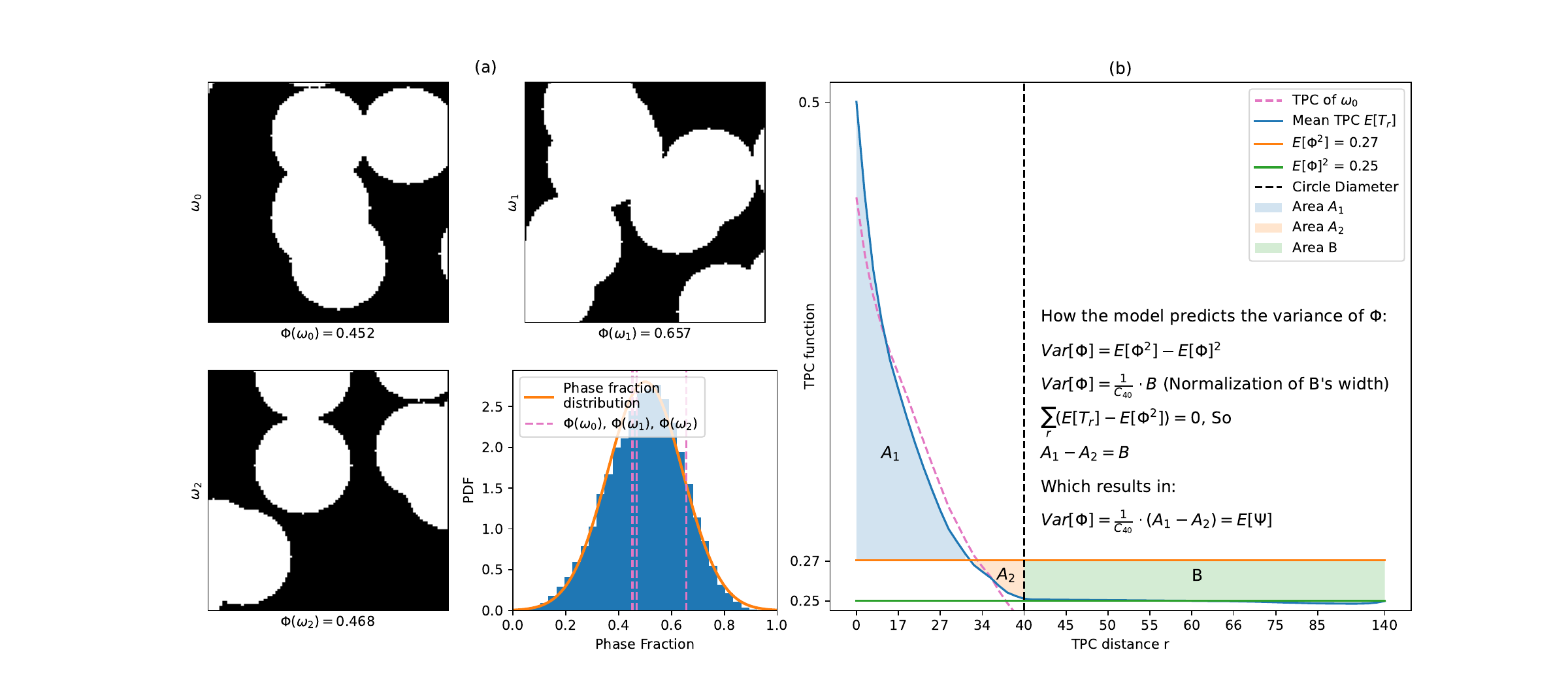}
    \caption{A visual explanation of the proof of Theorem \ref{theorem: expectation is variance}, in the simple case of randomly placed circles on a plane (overlapping permitted). On the left (a), there are three randomly generated $100^2$ images of randomly placed circles with diameter $40$. The histogram in the bottom right of (a) shows the measured phase fraction of $1000$ images that are generated in this way, and that in expectation the phase fraction will be $\text{E}[\Phi]=0.5$. The right plot (b) shows the TPC of image $\omega_0$ in dashed pink, and the expectation of the TPC over the $1000$ image samples, $\text{E}[T_r]$, for different radii lengths in blue. The entire $\omega_0$ TPC function is not presented since after $r=37$, the TPC is lower than 0.24. Because the circles are independently placed on the plane, after $r_0=40$, the probability of seeing the same phase at both ends of a vector is $\text{E}[\Phi]^2$, as shown in the plot. The longest distance of a vector in a $100^2$ image is $100\cdot\sqrt{2}$, hence the end distance of $140$. The variance of phase fraction for this simple case for image size $100^2$ is $0.02$ and standard deviation $0.141$.} 
    \label{fig:circles explanation}
\end{figure}

\subsubsection{Prediction of the Characteristic Length Scale}

For the characteristic length scale $a_n$ (the length scale of the Integral Range $A_n=a_n^n$) presented in Equation (\ref{eq:Var as a function of CLS}), if solved for $a_n$ we get

\begin{equation}\label{eq:CLS as a function of variance}
    a_n=\left(\frac{|X|\cdot \text{Var}[\Phi]}{\phi(1-\phi)}\right)^{\frac{1}{n}}
\end{equation}

Combining Equations (\ref{eq:Psi definition}) and (\ref{eq:CLS as a function of variance}), we can define a new random variable $\Tilde{a}_{n,X}:\Omega_X\rightarrow\mathbb{R}$ that is the resulting $a_n$ prediction for a single image:

\begin{equation}\label{eq:model CLS prediction}
    \Tilde{a}_{n,X}(\omega)=\left(\frac{|X|\cdot\Psi_X(\omega)}{\Phi(\omega)(1-\Phi(\omega))}\right)^{\frac{1}{n}}
\end{equation}

\subsection{Data-driven Understanding of the model error using the MicroLib library}

The MicroLib microstructure library \cite{kench2022microlib} is a diverse library of microstructures based on the DoITPoMS micrograph library \cite{barber2007doitpoms}. The dimensionality expansion from 2D to 3D in MicroLib was created using SliceGAN \cite{kench2021generating}. SliceGAN is trained on a single 2D micrograph, and can then generate arbitrarily large volumes of the microstructure, derived from its fully (transpose) convolutional architecture. This ability allows us to study the representativeness of these microstructures as we can generate as many and as large samples as we need for a statistical analysis.

The MicroLib dataset consists of 87 diverse entries. In general, while grayscale micrographs may contain more detailed features, segmented n-phase images are commonly required for quantitative analysis. This includes the majority of modelling approaches, as well as extraction of metrics such as phase fraction, surface area, feature morphology and tortuosity. For this reason, we constrain our datasets to the 77 segmented 2-phase materials. 

To validate the model, we wish to understand how well our model approximates the image phase fraction size-dependent variance of the MicroLib materials. To do this, instead of predicting the variance, we predict the standard deviation (std) \begin{equation}\label{eq:predicting std}
\Tilde{\sigma}_{n,X}(\omega):=\sqrt{\Psi(\omega)}=\sqrt{\frac{1}{C_{r_0}}\cdot\sum_{\substack{r \\ |r|\leq r_0}}\frac{|X_r|}{|X|}\left(T_r(\omega)-\Phi^2(\omega)\right)}
\end{equation} 

 of the phase fraction $\sigma_n=\sqrt{\text{Var}[\Phi]}$. We do this because the standard deviation directly relates to the standard error used in the confidence interval for the phase fraction presented in Equation (\ref{eq:inherent error}). However, 
unlike the proof of Theorem \ref{theorem: expectation is variance}, where we have used the linearity of the expectation, it is not initially clear how close the expected value of our prediction $\text{E}[\Tilde{\sigma}]$ is to the standard deviation $\sigma[\Phi]$. Predicting the standard deviation does not produce a clean result as in Theorem \ref{theorem: expectation is variance}, since $$\sigma[\Phi]=\sqrt{\text{Var}[\Phi]}=\sqrt{\text{E}[\Psi]}\neq \text{E}[\sqrt{\Psi}]=\text{E}[\Tilde{\sigma}]$$
But, using a known result of the Taylor expansion of $\sqrt{\Psi}$ up to second order around the expected value $\text{E}[\Psi]$ we get $$\sigma[\Phi]=\sqrt{\text{Var}[\Phi]}=\sqrt{\text{E}[\Psi]}\approx \text{E}[\sqrt{\Psi}] +\frac{\text{Var}[\Psi]}{8\cdot \text{E}[\Psi]^{1.5}}\approx \text{E}[\sqrt{\Psi}]=\text{E}[\Tilde{\sigma}]$$
where the second equality follows from Theorem \ref{theorem: expectation is variance} and the second approximation follows from the fact that the variance of $\Psi$ is very small compared to the expectation of $\Psi$. So $\text{E}[\Tilde{\sigma}]\approx\sigma[\Psi]$ and assures us to use this approximation.

\subsection{Model uncertainty}

\begin{figure}[!h]
    \centering
    \includegraphics[width=\textwidth]{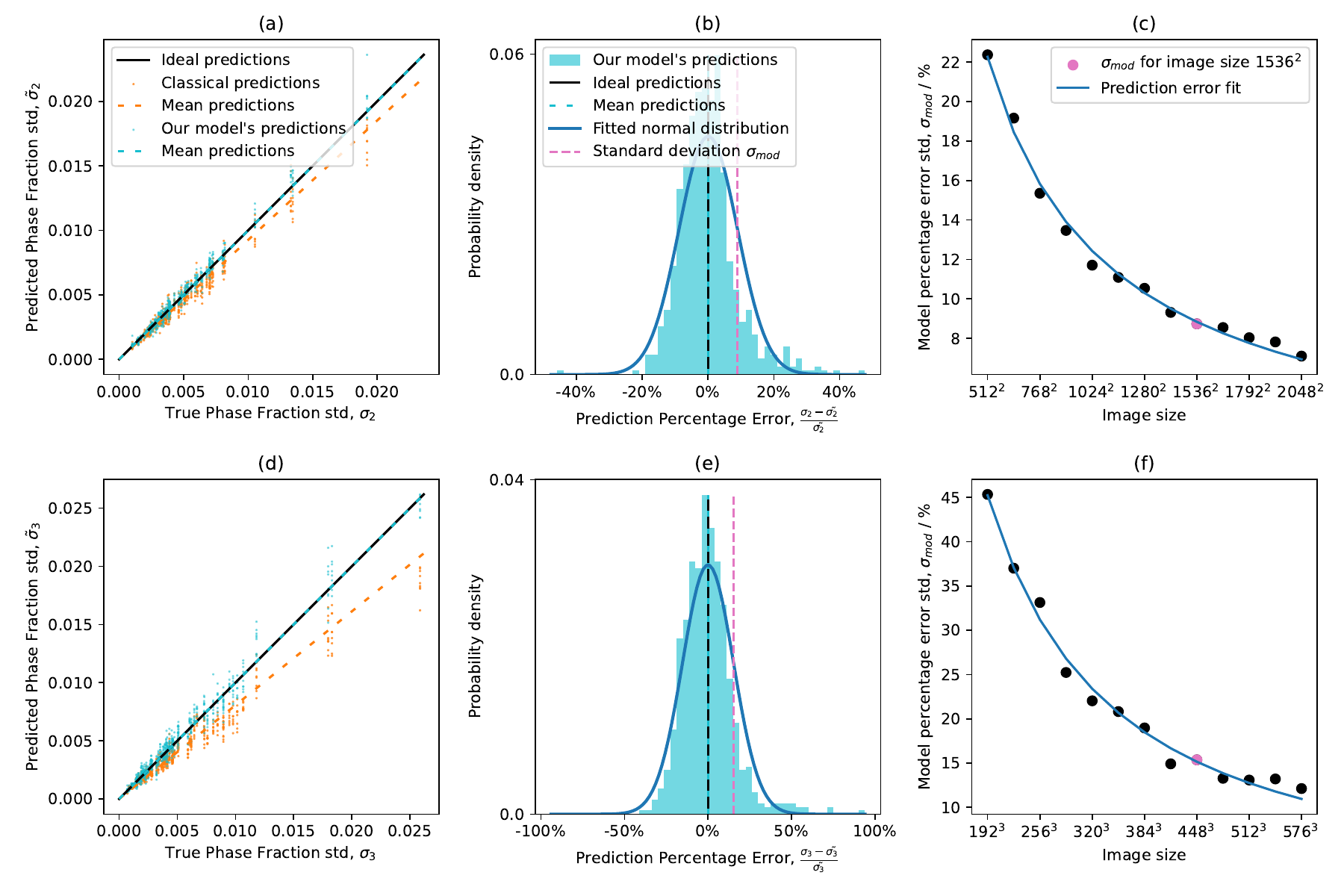}
    \caption{Determining our model's standard deviation prediction error. The top row refers to 2D images and the bottom to 3D images. Left is the prediction of the standard deviation $\Tilde{\sigma}_n$ for a specific image size ($1536^2$ for 2D in (a) and $448^3$ for 3D in (d)), and the real standard deviation $\sigma_n$ for the 77 different microstructures in MicroLib. There are 10 predictions for every microstructure for 10 different randomly generated images, hence the vertical stripes in the scatter plot, totalling in 770 $(\sigma_n,\Tilde{\sigma}_n)$ pairs. For every image, the prediction from our model is shown in teal, while the prediction from the classical subdivision method \cite{kanit2003determination} is shown in orange. The middle plots (b) \& (e) show the prediction percentage error for the plot on the left, including a fitted normal distribution. Using this fit, the standard deviation of this normal distribution, $\sigma_\text{mod}$, is found. The plots on the right, (c) \& (f), show the decay in model error (specifically, in the standard deviation of these fitted normal distributions) as the image size grows larger, with the reference for the image sizes presented in the left and middle plots is highlighted in pink. These error fits on the right are used for the final representativity predictions presented in Section \ref{subsec:Total error estimation}.}
    
    \label{fig:model error}
\end{figure}

To understand the accuracy of our model in predicting $\sigma_n$, the variation of $\Tilde{\sigma}_{n}$ is needed. The size and diversity of the MicroLib \cite{kench2022microlib} dataset allow us to explore this relationship and reach statistically significant conclusions. Figure \ref{fig:model error} shows the data-driven quantification of the model error in predicting the standard deviation in phase fraction $\Tilde{\sigma}_n$.

Specifically, we calculate $\Tilde{\sigma}_{n}$ based on a single image using Equation (\ref{eq:predicting std}) and calculate $\sigma_n$ using a statistical analysis of thousands of images generated with SliceGAN (referred to as the `true' standard deviation). Based on the predicted phase fraction standard deviation $\Tilde{\sigma}_{n}$, the method described above gives standard deviation predictions that correlate well with the real standard deviations. However, as expected, there is some spread, as shown in Figure \ref{fig:model error}.  This additional model error must be incorporated into the final prediction to avoid overconfident predictions. For this error quantification, the percentage error of the prediction from the true value was calculated for every image using the equation $\text{PE}(\Tilde{\sigma}_n, \sigma_n)=\frac{\sigma_n-\Tilde{\sigma}_n}{\Tilde{\sigma}_n}$.  

The predictions of our model are compared to the predictions of the `classical' method of finding the standard deviation of a single image through fitting the characteristic length scale as shown in the middle column of Figure \ref{fig:representativity introduction}, with its adaptation for a single image use as explained in the supplementary information. In contrast to Figure \ref{fig:representativity introduction}, where thousands of large images are sampled for finding the standard deviation of every image size, using this method on a single image (using subdivisions to smaller images) leads to significant under-predictions of the variance as can be seen in the left column of Figure \ref{fig:model error}. In comparison, our model's prediction percentage error mean in absolute value, is under $0.1\%$ for all image sizes. This disparity is further explored in the Results section.

The error model quantification 
process ends by finding the prediction error fit as shown in plots (c) and (f) in Figure \ref{fig:model error}. For each image size, a normal distribution is fitted to the prediction error histogram, and the standard deviation for this error distribution is found, $\sigma_\text{mod}$. The decay in standard deviation for different image sizes is then found, to be further used in the final prediction as explained in section \ref{subsec:Total error estimation}.

\subsection{Final representativity prediction}\label{subsec:Total error estimation}

\subsubsection{Predicting the confidence interval}\label{sec:predicting the confidence level}

\begin{figure}[!h]
    \centering
    \includegraphics[width=1\linewidth]{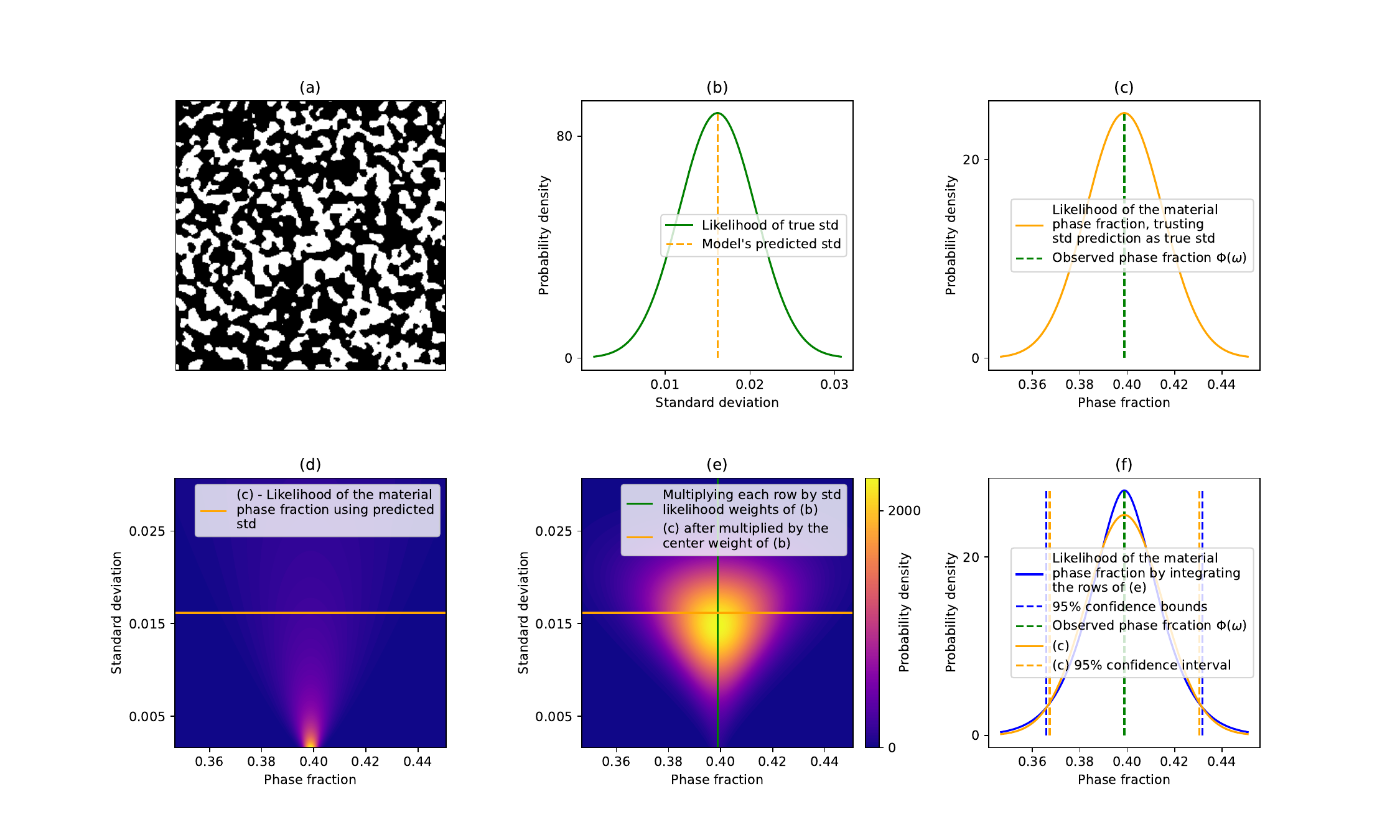}
    \caption{The final model confidence interval prediction, accounting for the error in standard deviation prediction. (a) An example of a $400^2$ image generated by a PoreSpy \cite{gostick2019porespy} generator, with phase fraction $\Phi(\omega)=0.398$. Our model's prediction of the material phase fraction standard deviation (of images of that size) is $\tilde{\sigma}=0.016$, shown by the orange dotted line in (b). The green line in (b) shows the model-predicted likelihood of the true standard deviation, given by $f_{\sigma} \sim \mathcal{N}(\tilde{\sigma}, \sigma_{\text{mod}}\cdot\tilde{\sigma})$. For image size $400^2$, the model estimates $\sigma_{\text{mod}}=27.9\%$ by observing Figure \ref{fig:model error}(c), so $f_{\sigma} \sim \mathcal{N}(0.016, 0.0045) $. (c) shows the likelihood of the bulk phase fraction assuming the predicted standard deviation is exact (and assuming normality of the phase fraction distribution)
, without accounting for the data-driven uncertainty \( \sigma_{\text{mod}} \). Each row of (d) is the phase fraction likelihood for a different standard deviation, with (c) in the center row. (e) shows (d) after multiplying each row by the likelihood from (b), forming the full 2D likelihood function of the bulk material's phase fraction. To obtain the resulting material's phase fraction likelihood, (f) shows the integration of (e) by its rows, applying the law of total probability as discussed in Section \ref{sec:predicting the confidence level}. (f) also shows (c) for comparison between the $95\%$ confidence intervals.}
    \label{fig:confidence bounds}
\end{figure}

In this section, we show how the model accounts for the uncertainty in the prediction of the phase fraction standard deviation, shown in Figure \ref{fig:model error}. After predicting the standard deviation $\tilde{\sigma}$ for a given image, the model also estimates the variability of this prediction, denoted $\sigma_\text{mod}$, using the data-driven function shown in Figure \ref{fig:model error}(c). These two quantities define a normal distribution for the predicted standard deviation. Specifically, we model the likelihood of the true standard deviation using a probability density function $$f_{\sigma}\sim \mathcal{N}(\tilde{\sigma}, \sigma_\text{mod}\cdot\tilde{\sigma})$$ as illustrated in Figure \ref{fig:confidence bounds}(b).

Let $B=[\Phi(\omega)-c,\Phi(\omega)+c]$ be the interval within which we want to estimate the probability of finding the material's true phase fraction. Then, by the law of total probability,

\begin{equation}\label{eq:prediction of confidence bounds}
    P(\phi\in B)=\int_{-\infty}^{\infty}P(\phi\in B\ |\ \sigma=x)\cdot f_\sigma(x)dx
\end{equation}

 Equation (\ref{eq:prediction of confidence bounds}) is the equation ultimately used for the prediction of the representativity confidence bounds of ImageRep. For example, for a $95\%$ confidence level, the value of $c$ is found in $B=[\Phi(\omega)-c,\Phi(\omega)+c]$ such that $P(\phi\in B)=0.95$. In practice, the limits of integration are  truncated to the range $\tilde{\sigma}\pm\tilde{\sigma}$, which captures more than $99.99\%$ of the probability mass of  $f_\sigma$ for most image sizes (since $\tilde{\sigma}>4\cdot\sigma_\text{mod}\cdot\tilde{\sigma}$).  A visual example of how Equation (\ref{eq:prediction of confidence bounds}) is built can be seen in Figure \ref{fig:confidence bounds}. 


\subsubsection{Predicting image size needed for a specified confidence interval}

ImageRep can predict the image size needed for a user specified confidence interval, i.e. what image size is needed for a particular deviation from the bulk material phase fraction, outlined in Algorithm \ref{alg:image size prediction}.

Following the notation in Algorithm \ref{alg:image size prediction}, given a confidence level $c$ and a required deviation level $d$, we apply a simple optimization algorithm to find the image size needed such that $P(\phi\in B)=c$ for $B=[(1-d)\cdot\Phi(\omega), (1+d)\cdot\Phi(\omega)]$ using Equation (\ref{eq:prediction of confidence bounds}) (remember $f_\sigma$ changes according to the image size).

\section{Validation Results}\label{sec: Validation}

\begin{figure}[h]
    \centering
    \includegraphics[width=1\linewidth]{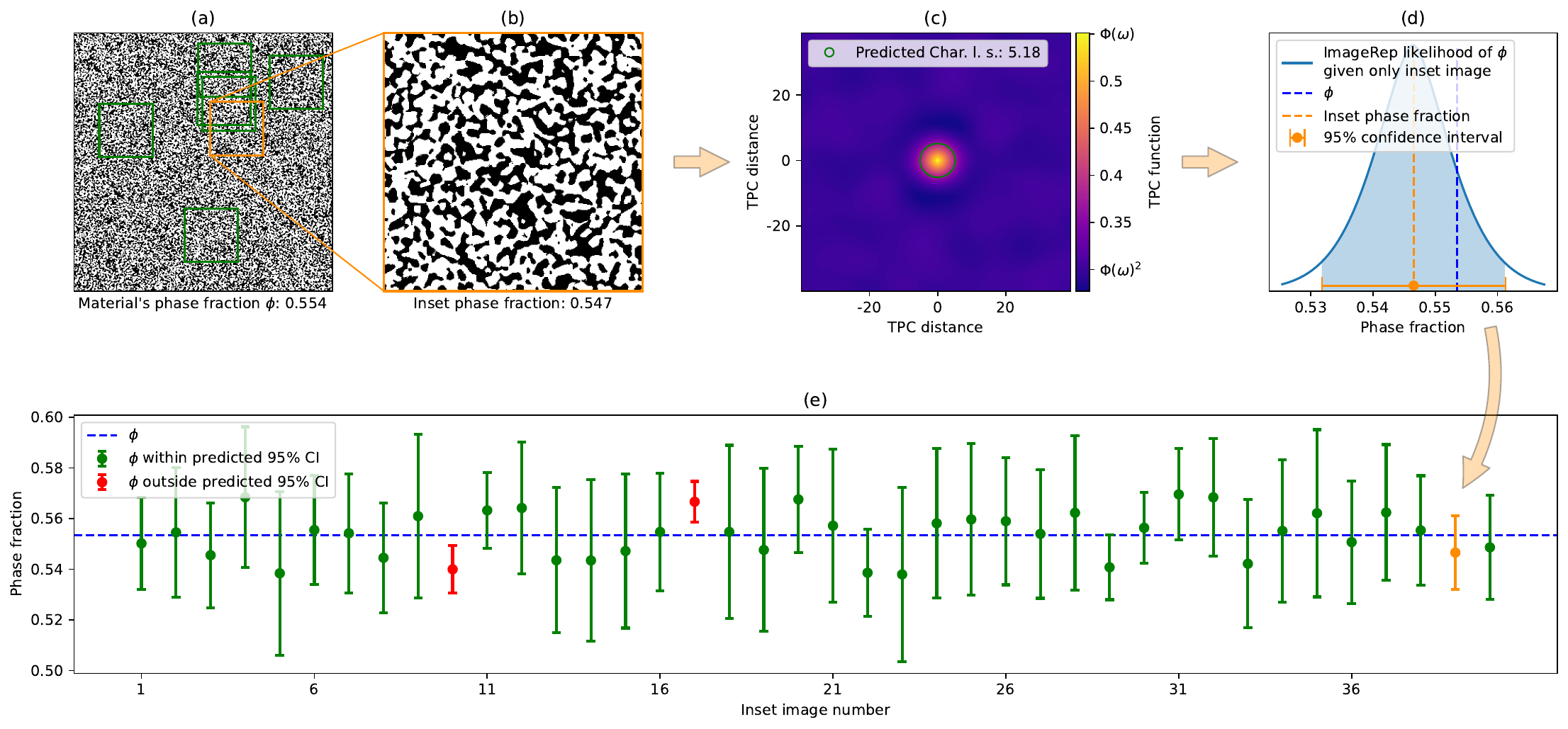}
    \caption{
    Overview of the ImageRep method validation process: For a given confidence level such as 95\%, the method needs to produce correct results 95\% of the time. Specifically, for the example of the validation of the SOFC anode material \cite{hsu2018mesoscale} shown in (a), a random image sample (b) is taken. The proposed method calculates the TPC function, and the corresponding characteristic length scale (c) and arrives at the predicted confidence interval (d) through the statistical framework presented in Section \ref{sec:model theory}. The bulk material's phase fraction $\phi$ \protect\footref{footnote:uncertainty in sofc large image}, should fall within the predicted confidence interval in $95\%$ of cases. (e) presents $40$ different confidence intervals for $40$ different random image samples taken from (a), showing that in $38/40=95\%$ of the cases, the material's phase fraction is within these intervals. the corresponding interval for image (b) is presented in inset image number 39.
    }
    
    \label{fig:model accuracy}
\end{figure}

\begin{table}[h]
\centering
\caption{Validation results of our method and the classical subdivision method, applied to the same images taken from the different materials.}
\label{table: validation}
\begin{tabular}{|ccc|}
\hline
\multicolumn{1}{|c|}{Method}             & \multicolumn{1}{c|}{\begin{tabular}[c]{@{}c@{}}Material's true phase fraction \\ is inside the predicted bounds\end{tabular}} & \begin{tabular}[c]{@{}c@{}}Absolute error from 95\% \\ confidence goal\end{tabular} \\ \hline
\multicolumn{3}{|c|}{PoreSpy materials}                                                                                                                                                                                                                             \\ \hline
\multicolumn{1}{|c|}{Subdivision method (2D)} & \multicolumn{1}{c|}{9238/10000 = 92.38\%}                                                                                     & 2.62\%                                                                              \\ \hline
\multicolumn{1}{|c|}{\textbf{ImageRep (2D)}}  & \multicolumn{1}{c|}{9442/10000 = 94.42\%}                                                                                     & \textbf{0.58\%}                                                                     \\ \hline
\multicolumn{1}{|c|}{Subdivision method (3D)} & \multicolumn{1}{c|}{914/1000 = 91.40\%}                                                                                       & 3.60\%                                                                              \\ \hline
\multicolumn{1}{|c|}{\textbf{ImageRep (3D)}}  & \multicolumn{1}{c|}{941/1000 = 94.10\%}                                                                                       & \textbf{0.90\%}                                                                     \\ \hline
\multicolumn{3}{|c|}{Solid Oxide Fuel Cell anode}                                                                                                                                                                                                                   \\ \hline
\multicolumn{1}{|c|}{Subdivision method (2D)} & \multicolumn{1}{c|}{967/1008 = 95.93\%}                                                                                       & 0.93\%                                                                              \\ \hline
\multicolumn{1}{|c|}{\textbf{ImageRep (2D)}}  & \multicolumn{1}{c|}{956/1008 = 94.84\%}                                                                                       & \textbf{0.16\%}                                                                     \\ \hline
\multicolumn{3}{|c|}{Targray PE16A battery separator}                                                                                                                                                                                                               \\ \hline
\multicolumn{1}{|c|}{Subdivision method (3D)} & \multicolumn{1}{c|}{142/200 = 71\%}                                                                                           & 24.00\%                                                                             \\ \hline
\multicolumn{1}{|c|}{\textbf{ImageRep (3D)}}  & \multicolumn{1}{c|}{191/200 = 95.5\%}                                                                                         & \textbf{0.50\%}                                                                     \\ \hline
\multicolumn{3}{|c|}{Celgard PP165 battery separator}                                                                                                                                                                                                               \\ \hline
\multicolumn{1}{|c|}{Subdivision method (3D)} & \multicolumn{1}{c|}{200/200 = 100\%}                                                                                          & 5.00\%                                                                              \\ \hline
\multicolumn{1}{|c|}{\textbf{ImageRep (3D)}}  & \multicolumn{1}{c|}{192/200 = 96\%}                                                                                           & \textbf{1.00\%}                                                                     \\ \hline
\multicolumn{3}{|c|}{All materials}                                                                                                                                                                                                                                 \\ \hline
\multicolumn{1}{|c|}{Subdivision method (2D)} & \multicolumn{1}{c|}{10205/11008 = 92.71\%}                                                                                    & 2.29\%                                                                              \\ \hline
\multicolumn{1}{|c|}{\textbf{ImageRep (2D)}}  & \multicolumn{1}{c|}{10398/11008 = 94.46\%}                                                                                    & \textbf{0.54\%}                                                                     \\ \hline
\multicolumn{1}{|c|}{Subdivision method (3D)} & \multicolumn{1}{c|}{1256/1400 = 89.71\%}                                                                                      & 5.29\%                                                                              \\ \hline
\multicolumn{1}{|c|}{\textbf{ImageRep (3D)}}  & \multicolumn{1}{c|}{1324/1400 = 94.57\%}                                                                                      & \textbf{0.43\%}                                                                     \\ \hline
\end{tabular}
\end{table}

Large samples of materials are required for validation to ensure that their phase fractions accurately represent the phase fraction of the entire material. This is not a concern with generated simulated materials, since there is no upper bound on the size of generated samples, except for computer resources. Therefore, for validation we used 50 different simulated materials and 3 real experimental large open-source microstructures of real materials, whose deviation from their true material phase fraction is predicted to be less than 1\%. 

The simulated materials were generated by PoreSpy \cite{gostick2019porespy}, a python library capable of generating a wide variety of microstructures. The name of the PoreSpy generators used are presented in the Supplementary Information. The three real materials are a solid oxide fuel cell (SOFC) anode \cite{hsu2018mesoscale} and two battery separators: a Targray PE16A separator and a Celgard PP1615 separator \cite{lagadec2016communication}. For the SOFC anode, we used the inpainting tool developed by Squires et al. \cite{squires2023artefact} to account for a small void defect, that is presented in the Supplementary Information.

All the validation data is `unseen' to the model, in the sense that it has not been used for the variability quantification step outlined in Figure \ref{fig:model error}. We put our confidence claim to the test - for a given confidence level of $95\%$, we claim that with $95\%$ confidence, the material's true phase fraction is within the model's predicted bounds. Therefore, if we test whether the material's true phase fraction is in the predicted bounds, ideally it should be in the bounds $95\%$ of the time. This workflow of the validation process is presented in Figure \ref{fig:model accuracy}.

Table \ref{table: validation} present the results of all 11008 material's samples used for validation, with the process explained in Figure \ref{fig:model accuracy}, with information broken down to the different materials used for validation. Our method and the classical subdivision method were applied on the same image samples taken for validation. 

The classical subdivision method is, to the best of our knowledge, the only representativity analysis approach currently used in practice. It is presented and applied in \cite{kanit2003determination}, and the details of our implementation are provided in the Supplementary Information. A clear visualization of the subdivision method applied to a single image is available in \cite{kanit2006apparent}. This method does not incorporate any information from the two-point correlation function. For instance, it does not utilize the uncorrelated distance estimate $r_0$, which could offer valuable insight and potentially enhance the basic approach.

The Supplementary Information also presents the results of our method without the final correction step introduced in Section \ref{subsec:Total error estimation}, which incorporates a data-driven variance prediction of the variance. These results highlight the importance of this step in achieving more accurate predictions.



\section{Discussion}\label{sec: Discussion}

The presented method, ImageRep, is a novel approach to the prediction of the representativity of phase fractions from a single micrograph or microstructure, providing a valuable tool for the advancement of the field of materials characterization. By leveraging the Two-Point Correlation function, our model directly estimates the phase fraction variance from a single image, thereby enabling robust confidence level predictions. This approach contrasts sharply with the traditional method that relies on large datasets to estimate representativity.

In the section that first introduced our definition of representativity, we reference the example of an SEM image of metal dewetting on the surface of ceramic substrate \cite{song2019unveiling}. After segmenting an image from this study using \cite{docherty2023samba}, we uploaded the result to imagerep.io and it gives us the following exact sentence ``The phase fraction in the segmented image is 0.405. Assuming perfect segmentation, the `ImageRep' model proposed by Dahari \textit{et al.} suggests that we can be 95.0\% confident that the material's phase fraction is within 4.6\% of this value (i.e. $0.405\pm0.019$).'' We provide this wording to ensure that the user quotes the results of the analysis appropriately. 

\subsection{Case studies: Impact of phase fraction uncertainty on battery performance}

Usseglio-Viretta \textit{et al.} \cite{usseglio2018resolving} shared a dataset of XCT images from a lithium-ion battery cathode material. The dataset is highly cited in the battery community partly due to the scarcity of alternative open data. When we applied our analysis to one of the published NMC cathode ($253^3$ voxels at $(398\ \text{nm})^3$ voxel size), we found that the $95\%$ confidence interval of the pore phase fraction was $0.365 \pm 0.015$. We imported the two bounding phase fraction values (i.e. 0.350 and 0.380) into the PyBaMM battery simulation software \cite{sulzer2021python}, with the default NMC cathode parameters that originate from \cite{chen2020development}, in order to investigate the impact of this uncertainty on device level performance. 

At low powers, the predicted capacity of these two cells was proportional to their respective solid phase fraction; however, at higher powers, the cells begin to diverge as the cell with a lower porosity becomes transport limited. So, while the less porous cell had $c.$ 11\% higher capacity at low power, it has $c.$ 9\% lower capacity at higher power. This swing in observed capacity is comparable to the difference in design between an “energy cell” and a “power cell”, suggesting that significantly larger images would be required in order to draw meaningful conclusions about the relationships between microstructure and performance in battery electrodes. The simulated discharge curves are presented in the Supplementary Information.

The analysis by Lu \textit{et al.} \cite{lu20203d} of the impact of microstructure on lithium-ion battery performance is batteries is one of
the most highly cited papers on microstructural imaging. The XCT dataset was $341\times342\times379$ voxels, with 126 nm voxel  size. However, it can be seen from Figure 3 of their article that the volumes used for simulation are only a few particles across. Analysis of their data (kindly provided by the authors) using ImageRep suggests that we can be 95\% confident that the NMC's phase fraction is within 10.4\% of this value (i.e. $0.406\pm0.042$). Furthermore, if they had wanted to be within 1\% of the true value they would have needed to image a volume \textit{c.} 120 times larger. This gives an insight into the significant representativity challenges faced by the battery community.

\subsection{Assessing image unrepresentativeness}

Instead of predicting the proximity of the image phase fraction to the material phase fraction, the model can also be used to predict how far it is from the material phase fraction, to predict how un-representative a sample is. 

This can be performed by observing the tails of the likelihood distribution, rather than focusing on the center of the likelihood distribution, as shown in Figure \ref{fig:model accuracy}. For example, the predicted phase fraction of graphite in micrograph 368 from DoITPoMS of Blackheart cast iron \cite{barber2007doitpoms} (after MicroLib \cite{kench2022microlib} segmentation), is within 18.4\% of the micrograph graphite phase fraction, with 95\% confidence ($0.261\pm0.048$). Conversely, instead of focusing on the bulk $95\%$ of the distribution, by examining the lower  and upper quartiles, we can predict that the bulk graphite phase fraction deviates by more than 6\% from the image phase fraction, with 50\% confidence. 

In other words, there is a 50\% probability (a coin toss) that the actual material phase fraction is lower than $0.245$ or higher than $0.277$, indicating that the image may be far from being representative. This approach provides a measure of how unreliable or unrepresentative a sample is in relation to the broader material's characteristics. The image and the analysis are presented in the Supplementary Information.

\subsection{Implications for Material Science}

Our method addresses a critical gap in the current literature: the lack of confidence analysis for metrics extracted from micrographs. Traditional approaches often necessitate substantial time and financial investments to collect numerous data repeats, posing a barrier to efficient material characterization. By reducing the data requirements, ImageRep offers a practical tool for material scientists and engineers, particularly when working with limited microstructural data.

The ability to predict phase fraction variation from a single image has several practical implications:
\begin{enumerate}
    \item \textbf{Enhanced Efficiency}: Researchers can achieve significant time and cost savings by eliminating the need for large datasets to determine confidence intervals.
    \item \textbf{Improved Accuracy}: Our method allows for materials manufacturers to have more precise control over phase fractions, which is crucial for optimizing materials for specific applications. Small changes in phase fraction can significantly impact material properties and performance, making accurate representativity predictions essential.
    \item \textbf{Wide Applicability}: The model's validation using diverse microstructures from open-source datasets demonstrates its efficacy across various material types, making it a versatile tool in the field.
\end{enumerate}

Additionally, the classical ``subdivision" method of single image variation prediction tends to under-predict the variance. We believe this is due to the poor fit of small subsamples to Equation (\ref{eq:variation as a finction of cls}), and the low variability of few large subsamples. Our TPC-based approach does not rely on Equation (\ref{eq:variation as a finction of cls}) directly, and accumulates more information from the image by observing all two-point correlations which captures the features variability. This seems especially important with anisotropic materials, as can be seen from the validation results presented in Table \ref{table: validation}, where the biggest gap in performance is observed in the 3D experimental separator materials which are anisotropic \cite{lagadec2016communication}.

The validation of our model using the MicroLib generators and testing using the PoreSpy generators and the experimental datasets confirms its robustness and accuracy. The prediction of the phase fraction standard deviation closely matches the true standard deviation. The model's error quantification, accounting for prediction variability, ensures that the confidence intervals are reliable.

Our method's accuracy is further validated through a rigorous test: predicting whether the true material's phase fraction falls within the model's confidence bounds for 95\% of the cases. The results show that the confidence bounds generated by our model are accurate close to 95\% of the time, underscoring the model's reliability. Its worth noting that the model was not validated for other specific confidence levels, but we see no reason to expect the results to differ. 

\subsection{Limitations and Future Work}

Despite its advantages, the use of the MicroLib database for quantifying the variance of the model output, which is the variance of the predicted phase fraction, has inherent limitations. The MicroLib dataset, while diverse, will not encompass all possible microstructural variations encountered in practice. Future work could expand the dataset to include more diverse microstructures, further enhancing the model's generalisability.

While we use the macro-homogeneity assumption  \cite{hill1984macroscopic, costanzo2005definitions} for our analysis by selecting $r_0$ as the threshold distance beyond which spatial correlations are assumed to vanish, we acknowledge that this introduces a simplifying assumption. This hard cutoff is necessary for the theoretical validity of our formulation, but in reality, correlations in random media often decay gradually rather than disappearing abruptly. More generally, such behavior is described by concepts like ergodicity \cite{shachar2023first} and mixing \cite{dentz2023mixing}, which characterize the weakening of statistical dependence with distance. Incorporating these frameworks, or relaxing the strict independence assumption, could offer a path toward improving the model's performance and further refining the theoretical connection between phase fraction variance and the two-point correlation function. In particular, systems with long-range order or suppressed fluctuations, such as hyperuniform materials, may require special consideration.

There are theoretical microstructures with infinite integral ranges, such as infinite fibrils \cite{dirrenberger2014towards}, that are beyond the scope of our current model. Addressing these cases, and specifically breaking the CLS into different angle orientations, would further strengthen the robustness of ImageRep in future work.

In addition, we would like to explore the application of a similar workflow to other material's metrics such as interfacial area and triple-phase boundary length, as well as transport phenomena such as the tortuosity factor. We would also like to explore how this method can be expanded to find the phase fraction representativity of n-phase materials for $n>2$ without treating the image as binary for each phase independently, which loses the interdependence of all phases.

It is also important to note that the uncertainty in representativity calculated using our approach assumes that the image has been perfectly segmented. However, in reality, it may be the case that the error of phase fraction due to poor segmentation is greater than the error due to low representativity in some scenarios. Additionally, uncertainties introduced by the imaging process itself, such as noise, contrast limitations, or resolution constraints, can impact the quality of the raw grayscale image, ultimately affecting the accuracy of the binary segmentation. These factors may introduce structural artifacts or misclassifications that propagate into the representativity prediction. A promising direction for future work would be to develop an end-to-end representativity model that integrates the effects of imaging, segmentation, and spatial analysis in a unified framework, enabling a more comprehensive and realistic assessment of uncertainty.

Finally, a general rule-of-thumb for representativity in terms of estimated feature size (i.e, average particle radius) would be highly desirable. This would allow practitioners to roughly decide on the image size needed whilst they are imaging, before performing the full scan, segmenting the resulting image and running ImageRep to find how representative their sample is. 


\section{Website}

To make the our model accessible to users across backgrounds (experimentalists, microscopists, theoreticians) and domains, we developed a simple webapp, \url{www.imagerep.io}. This app allows users to drag and drop multi-phase segmented microstructures in \texttt{.png}, \texttt{.jpg} or \texttt{.tiff} formats, select the desired phase and desired representativity confidence level. The confidence bounds are then returned, together with an interactive scale bar, which the users can change and observe how the confidence level affects the representativity estimation. The users can also receive the amount of additional measurements needed to reach smaller desired confidence bounds. We hope this ease of access and ease of use encourages more widespread reporting of metric representativity.

\section{Code availability}

The codes used in this manuscript and for creating the website are available at \url{https://github.com/tldr-group/ImageRep}

\section*{Acknowledgements}

This work was supported by funding from the Imperial President's PhD Scholarships received by A.D., funding from the EPSRC Faraday
Institution Multi-Scale Modelling project (https://faraday.ac.uk/;
EP/S003053/1, grant number FIRG003) received by S.K., and the EPRSC and SFI Centre for Doctoral Training in Advanced Characterisation of Materials (EP/S023259/1) received by R.D.. We also thank Isaac Squires for his initial research on this subject, Simon Daubner for his valuable comments as well as the rest of the tldr group.

\section*{Author contributions}

A.D., S.K. and S.J.C. conceptualised the project. A.D. designed and developed the code for ImageRep, built the mathematical model, performed the statistical analysis and drafted the manuscript. R.D. built the webapp. R.D. and S.K. revised parts of the code for ImageRep. R.D., S.K. and S.J.C. made substantial revisions and edits to all sections of the draft manuscript.

\section*{Competing interests}
The authors declare no competing interests.

\bibliographystyle{abbrv}
\bibliography{main.bib}
\end{document}


\title{ImageRep: Predicting microstructural representativity from a single image - Supplementary Information}

\author[1]{{Amir Dahari}
 \ }

\author[1, 2]{{Ronan Docherty}
 \ }

\author[1]{{Steve Kench}
 \ }

\author[1]{Samuel J. Cooper \ \ }

\affil[1]{{\textit{\footnotesize Dyson School of Design Engineering, Imperial College London, London SW7 2DB}}}
\affil[2]{{\textit{\footnotesize Department of Materials, Imperial College London, London SW7 2AZ}}}

\maketitle

\tableofcontents

\newpage

\section{Calculating periodic TPC using FFT}

The proof of the connection between the periodic Two-Point Correlation function (TPC) and the Discrete Fourier Transform (DFT) can be found in \cite{adams2012microstructure} (chapter 12.3.1). We present it here using our notation, and expanding on some calculations for clarity. The fact that the TPC can be calculated using DFT allows for a faster calculation using the Fast Fourier Transform (FFT) algorithm. If the image size is $|X|$, then the number of calculations for calculating the TPC drops from an order of $|X|^2$ to an order of  $|X|\cdot log(|X|)$. This is a dramatic change, for example if $X$ is a $1000^3$ volume, the number of computations drop from $1\cdot10^{18}$ to $2\cdot10^{10}$.

From this point forward, we drop the image $\omega$ notation for simplicity, as all calculations are made on a specific image. As in \cite{adams2012microstructure}, we provide the proof in 1D, and the extension to higher dimensions is straight-forward. In Section 3.1, the non-periodic TPC was defined. Here, we work with the periodic TPC, which we define by $T^p_{r,X}$ instead of $T_{r,X}$, where the multiplications are ``wrapped around" and the indices are indexed mod image length. For example in the simple case of 1D, 

\begin{equation}\label{eq:two-point correlation definition}
    T^p_{r,X}=\frac{1}{|X|}\sum_{x=0}^{|X|-1} B_x\cdot B_{x+r}
\end{equation}

and if $x+r\geq|X|$, then $x+r=x+r-|X|$ (mod $|X|$).

The DFT of order $0\leq k\leq |X|-1$ for the sequence of numbers $\boldsymbol{B}=\{B_0,B_1,\dots B_{|X|-1}\}$ is

\begin{equation*}
    \mathcal{F}_k(\boldsymbol{B})=\sum_{x=0}^{|X|-1} B_x\cdot e^{\frac{-2\pi xk}{|X|}} 
\end{equation*}

On the other hand, if we define the sequence of all TPC's as $\boldsymbol{T^p_X}=\{T^p_{0,X},T^p_{1,X},\dots,T^p_{|X|-1,X}\}$, then

\begin{equation*}
    \mathcal{F}_k(\boldsymbol{T^p_{X}})=\sum_{y=0}^{|X|-1}\frac{1}{|X|}\sum_{x=0}^{|X|-1} B_x\cdot B_{x+y}\cdot e^{\frac{-2\pi yk}{|X|}}=
\end{equation*}

\begin{equation*}
    =\frac{1}{|X|}\sum_{y=0}^{|X|-1}B_x\sum_{x=0}^{|X|-1} \cdot B_{x+y}\cdot e^{\frac{-2\pi yk}{|X|}} 
\end{equation*}

If we set $z=y+x$, then

\begin{equation*}
    \mathcal{F}_k(\boldsymbol{T^p_{X}})=\frac{1}{|X|}\sum_{y=0}^{|X|-1}B_x\sum_{z=y}^{y+|X|-1} \cdot B_{t}\cdot e^{\frac{-2\pi (z-x)k}{|X|}} =
\end{equation*}

and since the indices are periodic, the summation of $z$ from $z=y$ to $z=y+|X|-1$ is the same as the summation of $z$ from $0$ to $|X|-1$. In addition to taking one exponent outside of the inner sum result in
\begin{equation*}
    \mathcal{F}_k(\boldsymbol{T^p_{X}})=\frac{1}{|X|}\sum_{y=0}^{|X|-1}B_x\cdot e^{\frac{2\pi xk}{|X|}}\sum_{z=0}^{|X|-1} \cdot B_{z}\cdot e^{\frac{-2\pi zk}{|X|}} =
\end{equation*}

so,

\begin{equation}\label{eq:fk of tx}
    \mathcal{F}_k(\boldsymbol{T^p_{X}})=\frac{1}{|X|}\mathcal{F}_k(\boldsymbol{B}
    ^*)\mathcal{F}_k(\boldsymbol{B}) 
\end{equation}

where $\boldsymbol{B}^*$ is the complex conjugate of $\boldsymbol{B}$. 

Equation (\ref{eq:fk of tx}) holds for every $0\geq k\geq |X|-1$, so also

\begin{equation}\label{eq:fk for all k}
    \frac{1}{|X|}\sum_{k=0}^{|X|-1}\mathcal{F}_k(\boldsymbol{T^p_{X}})\cdot e^{\frac{2\pi rk}{|X|}}=\frac{1}{|X|^2}\sum_{k=0}^{|X|-1}\mathcal{F}_k(\boldsymbol{B}
    ^*)\cdot\mathcal{F}_k(\boldsymbol{B})\cdot e^{\frac{2\pi rk}{|X|}}
\end{equation}

The left hand side is exactly the periodic TPC with distance $r:\ T^p_{r,X}$ which follows from the inverse DFT definition, and the right hand side is 

\begin{equation}\label{eq:final FFT TPC connection}
    T^p_{r,X}=\frac{1}{|X|}\mathcal{F}_r^{-1}\left(\boldsymbol{\mathcal{F}}(\boldsymbol{B}
    ^*)\cdot\boldsymbol{\mathcal{F}}(\boldsymbol{B})\right)
\end{equation}

for the sequence 

$$\boldsymbol{\mathcal{F}}(\boldsymbol{B}
    ^*)\cdot\boldsymbol{\mathcal{F}}(\boldsymbol{B})=\left(\mathcal{F}_0(\boldsymbol{B}^*)\cdot\mathcal{F}_0(\boldsymbol{B}), \dots,\mathcal{F}_{|X|-1}(\boldsymbol{B}^*)\cdot\mathcal{F}_{|X|-1}(\boldsymbol{B})\right)$$

Therefore, calculating the TPC requires calculating the DFT (or the inverse DFT) of the image only 3 times using the FFT algorithm using Equation \ref{eq:final FFT TPC connection}. The extension to 2 and 3 dimensions is straight-forward - simply more summations in each equation.

\section{Generating random spheres on a plane with a given phase fraction expected value}

In Figure 2 of the main paper, a method for creating circles on a plane is presented. Importantly, given a desired phase fraction, the method generates circles such that the expected value of a random image matches the given phase fraction.

Given a desired phase fraction $\phi$, circle radius $r$ and image size $l\times l$, the algorithm for creating a random image is:

\begin{enumerate}
    \item Create a $(l+2r)\times(l+2r)$ image of independently uniform random numbers between $0$ and $1$.
    \item Every pixel can be the center of a circle (the circles can overlap), a pixel is independently chosen to be the center of a circle if the probability value in that pixel is under $$1 - (1 - \phi)^{\frac{1}{\pi r^2}}$$ \label{probability of a pixel to be the center of a circle}
    \item Create circles around successful pixels ($1$ values) and all other values turn to $0$.
    \item Crop the center $l\times l$ image.
\end{enumerate}

To prove that the expected value of the phase fraction from this random generation is $\phi$, we need to show that the expected value of each pixel is $\phi$. A pixel takes the value 1 if it belongs to any circle, whether as the center or part of another circle. The probability that a pixel does not belong to any circle equals the probability that none of the surrounding $\pi r^2$ pixels are chosen as circle centers. The probability that a single pixel is not the center of a circle is $$(1 - \phi)^{\frac{1}{\pi r^2}}$$ Since the selection of centers is independent, the probability that a pixel has a value of 0 (i.e., does not belong to any circle) is: $$\left((1 - \phi)^{\frac{1}{\pi r^2}}\right)^{\pi r^2}=1-\phi$$

And the complement probability of a pixel to be $1$ is $\phi$. The last crop step ensures that all pixels have the same probability to be in a circle.

\section{Finding $r_0$ in practice}

When analyzing a single image of a macro-homogeneous material, the threshold distance $r_0$, as discussed in Section 3 of the main paper, cannot be determined in advance. This threshold represents the distance beyond which data points are (almost) uncorrelated.

In practice, we select $r_0$ from a set of distances in increments of 100 pixels, i.e., $r_0 \in {100, 200, ...}$. We define $r_0$ as the smallest distance such that, for a ring in the TPC function (e.g., between 200 and 300 TPC distances), the TPC values deviate from $\phi^2$ by more than $0.05 \cdot (\phi - \phi^2)$ in fewer than $3\%$ of the cases. This ensures that the signal to noise ratio is high and that we integrate valuable information from the TPC function. The increments of $100$ are chosen for consistency measurements between different images of the same material. 

Understanding the impact of $r_0$ on the final representativity prediction may be useful for improving the method in a further studies. Its selection may be improved by applying a median filter to the TPC curve, which could help reduce noise and yield a more robust estimate. Exploring such techniques could make the identification of $r_0$ more stable, especially when analyzing single images.

\section{Variation prediction using periodic TPC}

To use the fast computation of the FFT for the TPC results, we changed the prediction of the variation to use periodic TPC. Specifically, using the same notation as in Section 3 in the main paper,

\begin{theorem}\label{theorem: periodic expectation is variance}
    For large enough $r_0$ and appropriate constant $C^p_{r_0}$, if we set $\Psi^p:\Omega\rightarrow\mathbb{R}$ to be
    $$\Psi^p(\omega)=\frac{1}{C^p_{r_0}}\cdot\sum_{r\leq r_0}\left(T^p_{r,X}(\omega)-\Phi^2(\omega)\right)$$ then $E[\Psi^p]=Var[\Phi]$.
\end{theorem}

\begin{proof}

Every periodic TPC can be divided into two parts, one part which is the non-periodic TPC and one part who "wraps" around and has long-range correlations (if the TPC vector is shorter than half the image size). I. e. if the complement periodic vector to $r$ is $r_p$, then

\begin{equation*}
    T^p_{r,X}=\frac{|X_r|}{|X|}\cdot T_{r,X}+\frac{|X|-|X_r|}{|X|}\cdot T_{r^p,X}
\end{equation*}

So

\begin{equation}\label{eq:beg of periodic proof}
    \sum_{r\leq r_0}\left(T^p_{r,X}-\Phi^2\right)=
\end{equation}

\begin{equation}
=\sum_{r\leq r_0}\left(\frac{|X_r|}{|X|}\cdot T_{r,X}+\frac{|X|-|X_r|}{|X|}\cdot T_{r^p,X}-\Phi^2\right)=
\end{equation}

\begin{equation}\label{eq:division to two parts}
=\sum_{r\leq r_0}\frac{|X_r|}{|X|}\cdot\left( T_{r,X}-\Phi^2\right)+\sum_{r\leq r_0}\frac{|X|-|X_r|}{|X|}\cdot\left( T_{r^p,X}-\Phi^2\right)
\end{equation}

In expectation, the left additive part of Equation (\ref{eq:division to two parts}) is exactly $C_{r_0}\cdot Var[\Phi]$ by the theorem in Section 3 in the main paper. So taking the expectation of both sides (Equations (\ref{eq:beg of periodic proof}) and (\ref{eq:division to two parts})) yields

\begin{equation}\label{eq:taking expectation periodic proof}
    \sum_{r\leq r_0}\left(E[T^p_{r,X}]-E[\Phi^2]\right)=
    C_{r_0}\cdot Var[\Phi]+\sum_{r\leq r_0}\frac{|X|-|X_r|}{|X|}\cdot\left( E[T_{r^p,X}]-E[\Phi^2]\right)
\end{equation}

We assume that the image is large enough such that $r_0$ is much smaller then half of the image size, which means that for all $r\leq r_0,\ r^p>>r_0$ and so $E[T_{r^p,X}]=E[\Phi]^2$ by the macro-homogeneous assumption presented in Section 3 in the main paper. Thus, since $Var[\Phi]=E[\Phi]^2-E[\Phi]^2$, Equation (\ref{eq:taking expectation periodic proof}) becomes

\begin{equation}\label{eq:almost there periodic proof}
    \sum_{r\leq r_0}\left(E[T^p_{r,X}]-E[\Phi^2]\right)=
    \left(C_{r_0}+\sum_{r\leq r_0}\frac{|X|-|X_r|}{|X|}\right)\cdot Var[\Phi]
\end{equation}
 and by setting 

 $$C_{r_0}^p=C_{r_0}+\sum_{r\leq r_0}\frac{|X|-|X_r|}{|X|}$$

 we reach the required result.
 
\end{proof}
\subsection{Other changes for enhanced stability results}

We've made two changes to the prediction presented by $\Psi^p$ presented in Theorem \ref{theorem: periodic expectation is variance}, which help stabilise the output predictions:

\begin{enumerate}
    \item instead of using $\Phi^2$ in the equation of Theorem \ref{theorem: periodic expectation is variance}, we normalise this value with the TPC values that are close to $r_0$. Specifically, If $$R_0=\{r\ :\ |r_0-10|\leq |r|\leq |r_0|\}$$ then this value is defined by
    $$\Tilde{T}_{r_0,X}=\frac{1}{|R_0|}\sum_{r\in R_0}T_{r,X}$$
    and the value used instead of $\Phi^2$ is a mean of these two values,
    $$\frac{\Phi^2+\Tilde{T}_{r_0,X}}{2}$$

    This small change might help with stability by decreasing the sensitivity of the predictions to the image phase fraction, but a further study into this stabilisation is needed.
    
    \item Another implementation for stability is by comparing the model's prediction to that of the one-image subdivision method. The one-image subdivision method is computationally lighter than our model and does not significantly impact computation speed. If the model's prediction is more than three times higher or less than one-third of the subdivision method's prediction, the TPC values are adjusted to be more positive or more negative, respectively. In practice, this adjustment is only applied to very small images.    
\end{enumerate}

\section{Validation information breakdown}

\begin{figure}
    \centering
    \includegraphics[width=1\linewidth]{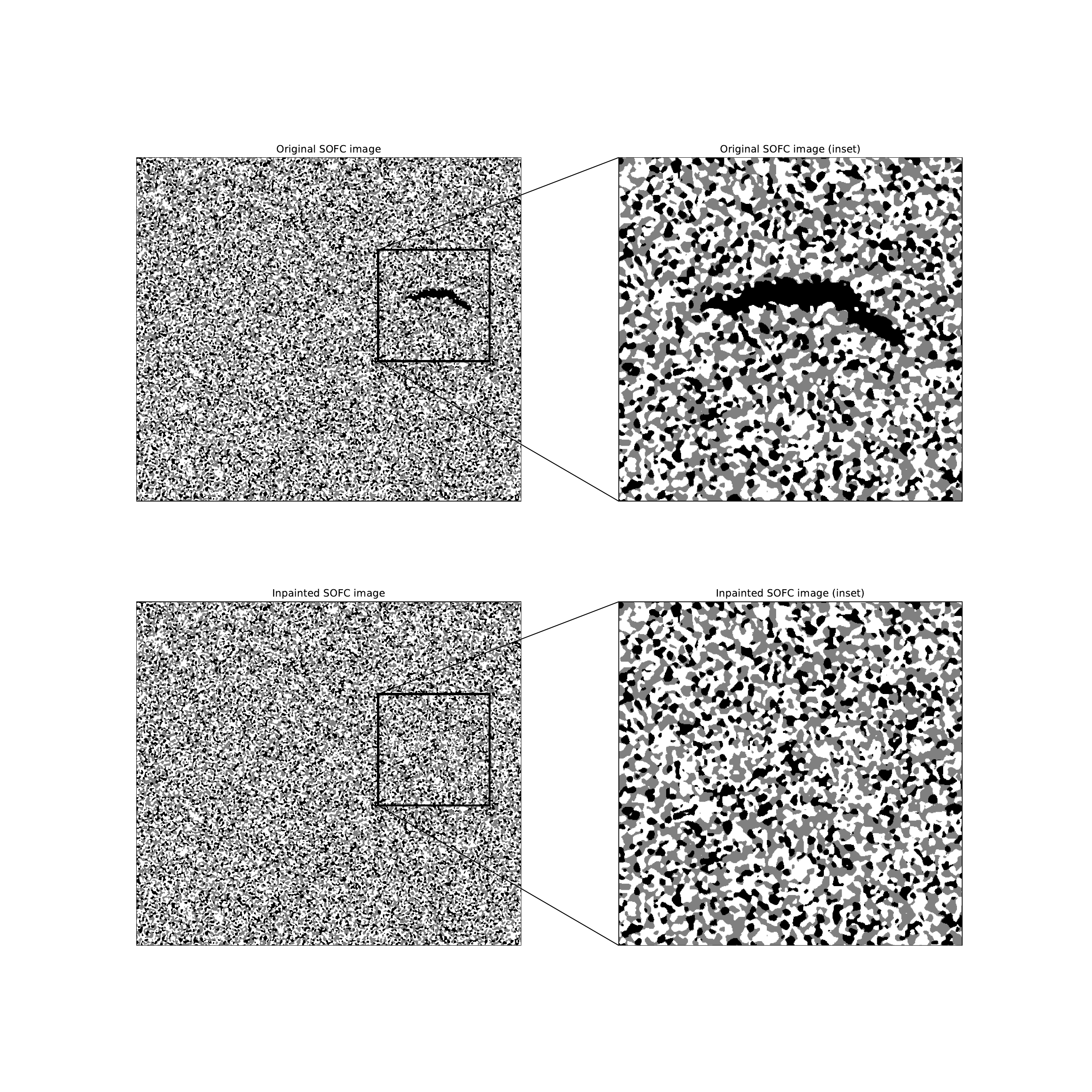}
    \caption{The origianl 2D slice (number 46) of the SOFC anode \cite{hsu2018mesoscale}, the inpainted image used for validation and the void defect inpainted area. We used \cite{squires2023artefact} for the inpainting.}
    \label{fig:inpainting}
\end{figure}

\subsection{PoreSpy generators used}

\begin{table}[!h]
\centering
\caption{Breakdown of generators used from PoreSpy \cite{gostick2019porespy}. }
\begin{tabular}{|cc|c|c|c|}
\hline
\multicolumn{2}{|c|}{Generator kind}                              & Param 1                                                                                                             & Param 2                & Param 3                \\ \hline
\multicolumn{1}{|c|}{\multirow{2}{*}{Fractal Noise}} & Param name & "frequency"                                                                                                         & "octaves"              & "mode"                 \\ \cline{2-5} 
\multicolumn{1}{|c|}{}                               & Option set & \begin{tabular}[c]{@{}c@{}}\{0.015, 0.025, 0.035, \\ 0.045, 0.055\}\end{tabular}                                    & \{2, 7, 12\}           & \{"simplex", "value"\} \\ \hline
\multicolumn{1}{|c|}{\multirow{2}{*}{Blobs}}         & Param name & "blobiness"                                                                                                         & "porosity"             &                        \\ \cline{2-5} 
\multicolumn{1}{|c|}{}                               & Option set & \begin{tabular}[c]{@{}c@{}}\{l/50, l/100, l/150, l/200, \\ l/250\} for l = average \\ length of image.\end{tabular} & \{0.1, 0.2, 0.3, 0.4\} &                        \\ \hline
\end{tabular}
\label{table: porespy generators}
\end{table}

Table \ref{table: porespy generators} shows the 50 different PoreSpy \cite{gostick2019porespy} generators used, that are all the combinations of parameters in the table. For example, a Fractal Noise generator was used with the parameters "frequency"=0.025, "octaves"=12 and "mode"="value". There are $5\cdot3\cdot2=30$ Fractal Noise generators and $5\cdot4=20$ Blobs generators.

\subsection{SOFC void defect inpainting}

The SOFC anode dataset \cite{hsu2018mesoscale} is a large 3D dataset of a 3-phase material. The dataset was used for 2D validation since one of the directional axis is small (the anode is not thick enough). For independence between the chosen slices, we picked $4$ slices with $23$ voxels distance between them (slice 46, 69, 92 and 115). This seems like a reasonable choice that would produce extra meaningful validation information given that the characteristic length scale of the different phases is between 7 and 9. The dataset has a local void defect that would harm the validation results. The defect is local and every sample from that region will have an extreme phase fraction that will be very different from the bulk phase fraction. We've decided to inpaint this region to not lose a large section from the image due to cropping the image. Figure \ref{fig:inpainting} shows how the inpainting was carried out to slice 46. 

\begin{table}[!h]
\centering
\caption{Breakdown of all validation results}
\begin{tabular}{|cccc|}
\hline
\multicolumn{1}{|c|}{Method type}                     & \multicolumn{1}{c|}{\begin{tabular}[c]{@{}c@{}}Material's true phase fraction \\ is inside the predicted bounds\end{tabular}} & \multicolumn{1}{c|}{Confidence goal} & Absolute error  \\ \hline
\multicolumn{4}{|c|}{PoreSpy materials}                                                                                                                                                                                                        \\ \hline
\multicolumn{1}{|c|}{Subdivision method (2D)}         & \multicolumn{1}{c|}{9238/10000 = 92.38\%}                                                                                     & \multicolumn{1}{c|}{95\%}            & 2.62\%          \\ \hline
\multicolumn{1}{|c|}{ImageRep only std (2D)}          & \multicolumn{1}{c|}{9404/10000 = 94.00\%}                                                                                     & \multicolumn{1}{c|}{95\%}            & 0.96\%          \\ \hline
\multicolumn{1}{|c|}{\textbf{ImageRep (2D)}}          & \multicolumn{1}{c|}{9442/10000 = 94.42\%}                                                                                     & \multicolumn{1}{c|}{95\%}            & \textbf{0.58\%} \\ \hline
\multicolumn{1}{|c|}{Subdivision method (3D)}         & \multicolumn{1}{c|}{914/1000 = 91.40\%}                                                                                       & \multicolumn{1}{c|}{95\%}            & 3.60\%          \\ \hline
\multicolumn{1}{|c|}{ImageRep only std (3D)}          & \multicolumn{1}{c|}{930/1000 = 93.00\%}                                                                                       & \multicolumn{1}{c|}{95\%}            & 2.00\%          \\ \hline
\multicolumn{1}{|c|}{\textbf{ImageRep (3D)}}          & \multicolumn{1}{c|}{941/1000 = 94.10\%}                                                                                       & \multicolumn{1}{c|}{95\%}            & \textbf{0.90\%} \\ \hline
\multicolumn{4}{|c|}{Solid Oxide Fuel Cell anode}                                                                                                                                                                                              \\ \hline
\multicolumn{1}{|c|}{Subdivision method (2D)}         & \multicolumn{1}{c|}{967/1008 = 95.93\%}                                                                                       & \multicolumn{1}{c|}{95\%}            & 0.93\%          \\ \hline
\multicolumn{1}{|c|}{ImageRep only std (2D)}          & \multicolumn{1}{c|}{943/1008 = 93.55\%}                                                                                       & \multicolumn{1}{c|}{95\%}            & 1.45\%          \\ \hline
\multicolumn{1}{|c|}{\textbf{ImageRep (2D)}}          & \multicolumn{1}{c|}{956/1008 = 94.84\%}                                                                                       & \multicolumn{1}{c|}{95\%}            & \textbf{0.16\%} \\ \hline
\multicolumn{4}{|c|}{Targray battery separator}                                                                                                                                                                                                \\ \hline
\multicolumn{1}{|c|}{Subdivision method (3D)}         & \multicolumn{1}{c|}{142/200 = 71\%}                                                                                           & \multicolumn{1}{c|}{95\%}            & 24.00\%         \\ \hline
\multicolumn{1}{|c|}{\textbf{ImageRep only std (3D)}} & \multicolumn{1}{c|}{189/200 = 94.5\%}                                                                                         & \multicolumn{1}{c|}{95\%}            & \textbf{0.50\%} \\ \hline
\multicolumn{1}{|c|}{\textbf{ImageRep (3D)}}          & \multicolumn{1}{c|}{191/200 = 95.5\%}                                                                                         & \multicolumn{1}{c|}{95\%}            & \textbf{0.50\%} \\ \hline
\multicolumn{4}{|c|}{PP165 battery separator}                                                                                                                                                                                                  \\ \hline
\multicolumn{1}{|c|}{Subdivision method (3D)}         & \multicolumn{1}{c|}{200/200 = 100\%}                                                                                          & \multicolumn{1}{c|}{95\%}            & 5.00\%          \\ \hline
\multicolumn{1}{|c|}{\textbf{ImageRep only std (3D)}} & \multicolumn{1}{c|}{192/200 = 96\%}                                                                                           & \multicolumn{1}{c|}{95\%}            & \textbf{1.00\%} \\ \hline
\multicolumn{1}{|c|}{\textbf{ImageRep (3D)}}          & \multicolumn{1}{c|}{192/200 = 96\%}                                                                                           & \multicolumn{1}{c|}{95\%}            & \textbf{1.00\%} \\ \hline
\multicolumn{4}{|c|}{All materials}                                                                                                                                                                                                            \\ \hline
\multicolumn{1}{|c|}{Subdivision method (2D)}         & \multicolumn{1}{c|}{10205/11008 = 92.71\%}                                                                                    & \multicolumn{1}{c|}{95\%}            & 2.29\%          \\ \hline
\multicolumn{1}{|c|}{ImageRep only std (2D)}          & \multicolumn{1}{c|}{10347/11008 = 94.00\%}                                                                                    & \multicolumn{1}{c|}{95\%}            & 1.00\%          \\ \hline
\multicolumn{1}{|c|}{\textbf{ImageRep (2D)}}          & \multicolumn{1}{c|}{10398/11008 = 94.46\%}                                                                                    & \multicolumn{1}{c|}{95\%}            & \textbf{0.54\%} \\ \hline
\multicolumn{1}{|c|}{Subdivision method (3D)}         & \multicolumn{1}{c|}{1256/1400 = 89.71\%}                                                                                      & \multicolumn{1}{c|}{95\%}            & 5.29\%          \\ \hline
\multicolumn{1}{|c|}{ImageRep only std (3D)}          & \multicolumn{1}{c|}{1311/1400 = 93.64\%}                                                                                      & \multicolumn{1}{c|}{95\%}            & 1.36\%          \\ \hline
\multicolumn{1}{|c|}{\textbf{ImageRep (3D)}}          & \multicolumn{1}{c|}{1324/1400 = 94.57\%}                                                                                      & \multicolumn{1}{c|}{95\%}            & \textbf{0.43\%} \\ \hline
\end{tabular}
\label{table:all validation results}
\end{table}

\subsection{Image sizes used for validation}

Different image sizes were used for validation. In 2D, for each PoreSpy generator, for each image size $\{600^2, 800^2, 1000^2, 1200^2, 1400^2\}$, $40$ images were randomly sampled from the generated materials, totalling in $200$ samples per material. For 3D, because of the longer computation time, for each image size $\{280^3, 310^3, 340^3, 370^3, 400^3\}$, $4$ images were randomly sampled from the generated materials, totalling in $20$ samples per material.

For the SOFC material, for each of the 3 phases, smaller images were sampled $\{400^2, 500^2, 600^2\}$ and also for the 3D separators $\{230, 260, 290, 320\}$ since the experimental microstructures are smaller than the PoreSpy simulated ones.

\subsection{Classical one image subdivision method}

For applying the subdivision method for a single image, we apply the same procedure as in \cite{kanit2003determination, kanit2006apparent}. The image is progressively divided into smaller subsamples, and for each set of subsamples of the same size, we calculate the standard deviation of their phase fractions.

After computing standard deviations for different sizes, the characteristic length scale (or integral range) is found by fitting the best integral range in Equation (2) presented in Section 2 of the main text. 

Specifically, the image is divided into ratios $2, 4, 8, ...$ while ratio 2 divides a 2D image into 4 and a 3D image into 8, ratio 4 divides a 2D image into 16 and a 3D image into 64, and so on. Very small subsamples such as 1, 2, 3 or 4 pixels/voxels are not used as their phase fractions are very unstable with very high standard deviations.

\begin{figure}[!h]
    \centering
    \includegraphics[width=1\linewidth]{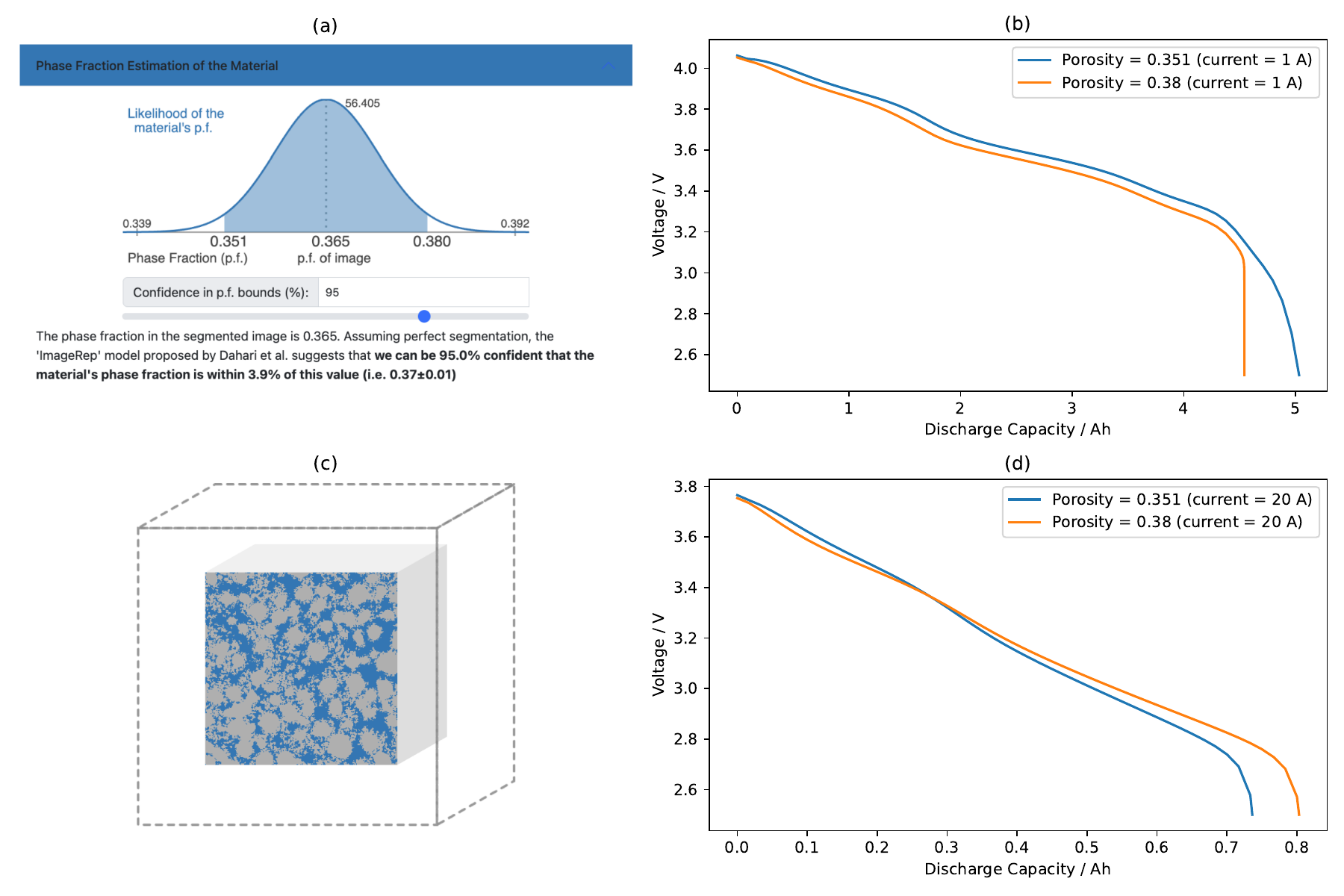}
    \caption{Impact of phase fraction uncertainty on battery performance. (a) presents the 95\% confidence interval of the bulk material phase fraction of the 2nd calendered microstructure from \cite{usseglio2018resolving}. (b) and (d) presents the discharge curves of two NMC cathodes using Pybamm \cite{sulzer2021python} with the default NMC cathode parameters that originate from \cite{chen2020development} by only changing the porosity values, the two bounding phase fraction values of the confidence interval (i.e. 0.351 and 0.380). (b) presents the discharge curves at low current ($\sim C/5$) and (d) at high current ($\sim 4C$). (c) Presents the volume that its phase fraction is predicted to be within only 2\% of the material phase fraction with 95\% confidence, instead of 3.9\% with the current volume, a $392^3$ voxels volume instead of $252^3$ voxels volume, at the same resolution.}
    \label{fig:SI discharge curves}
\end{figure}

\subsection{Validation including only std prediction step}

Table \ref{table:all validation results} shows all of the validation results. The label 'only std' is used to indicate that the prediction was made solely by estimating the standard deviation of the sample using the Two-Point Correlation function, without the data-driven correction step presented in the end of Section 3. The results presented here also show why this step helps prediction accuracy.

\section{Discussion SI}

\subsection{Impact of phase fraction uncertainty on battery performance}

Figure \ref{fig:SI discharge curves} shows the impact of phase fraction uncertainty on battery performance, and how the model predict the image size needed for a specific smaller uncertainty. 

\subsection{Assessing image unrepresentativeness}

\begin{figure}[!h]
    \centering
    \includegraphics[width=1\linewidth]{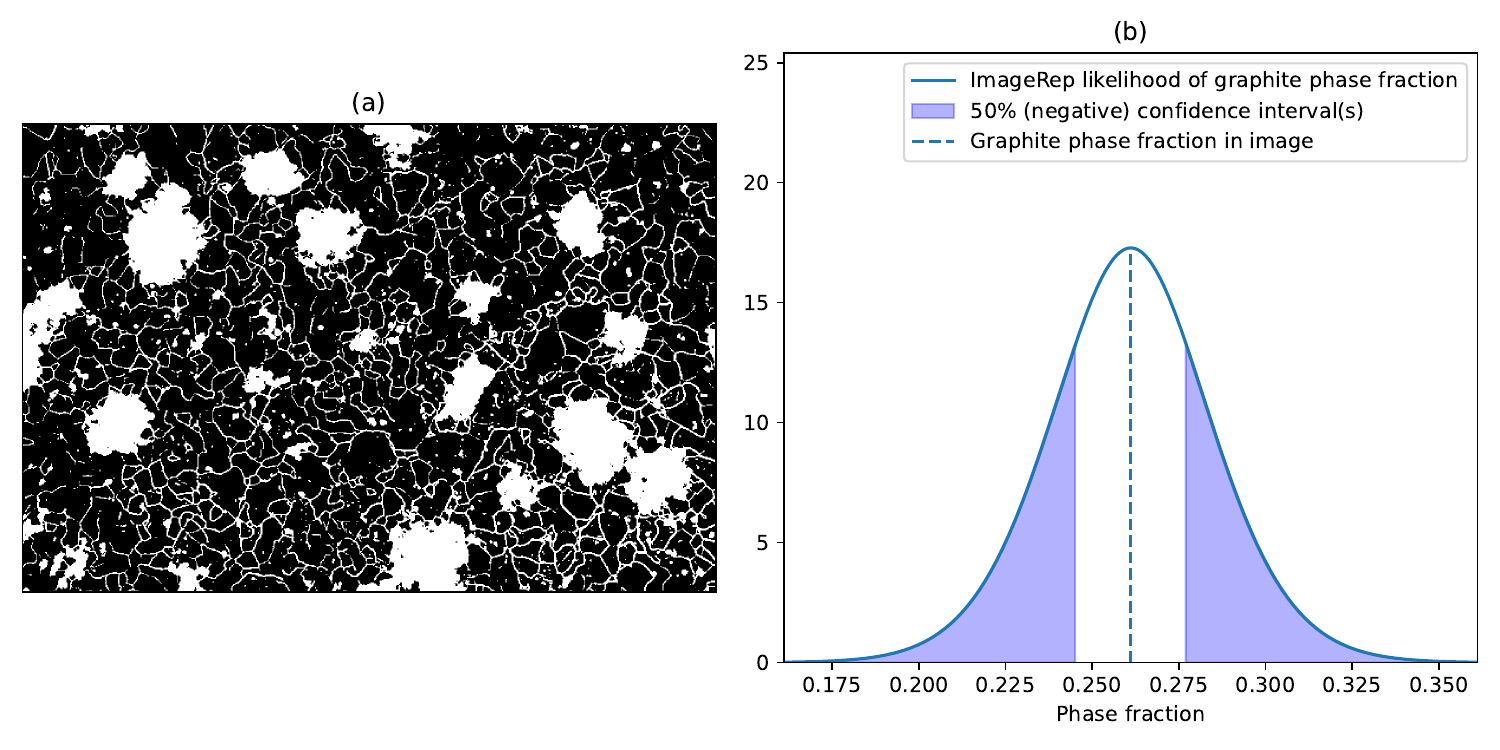}
    \caption{(a) Micrograph 368 from DoITPoMS of Blackheart cast iron \cite{barber2007doitpoms} (after microlib \cite{kench2022microlib} segmentation), showing the graphite phase along with ferrite grain boundaries. (b) Presents the 50\% tails of the likelihood distribution of the true phase fraction of the material. There is a 50\% probability (a coin toss) that the actual material phase fraction deviates by more than 6\% from the apparent phase fraction in the image, indicating that the image may be far from being representative of the material phase fraction.}
    \label{fig:enter-label}
\end{figure}
 

\newpage

\newpage

\newpage

\bibliography{main.bib}